\documentclass[12pt]{article}

\usepackage[T1]{fontenc}
\usepackage{lmodern}

\usepackage{setspace}
\usepackage[margin=1.25in]{geometry}

\usepackage[dvipsnames]{xcolor}
\usepackage{pdfpages}
\usepackage{float}
\usepackage{multirow}
\usepackage{caption}
\usepackage{subcaption}
\usepackage{epigraph}

\AtBeginDocument{%
  \everymath{\color{Violet}}%
  \everydisplay{\color{Violet}}%
}

\usepackage{mathtools}
\usepackage{amssymb}
\usepackage{amsthm}
\usepackage{amsfonts}
\usepackage{aliascnt}
\usepackage{accents}
\usepackage{dutchcal}
\usepackage{bbm}

\usepackage{enumitem}
\usepackage{tikz}
\usepackage{tikz-cd}
\usetikzlibrary{calc,shapes,arrows}
\usepackage{tcolorbox}
\usepackage{listings}

\usepackage[round]{natbib}
\DeclareMathOperator{\supp}{supp}

\DeclareMathOperator*{\argmax}{arg\,max}
\DeclareMathOperator*{\argmin}{arg\,min}

\newcommand{\E}{\mathbb{E}}
\newcommand{\R}{\mathbb{R}}

\newcommand{\1}{\mathbbm{1}}

\newtheorem{theorem}{Theorem}[section]

\newaliascnt{proposition}{theorem}
\newtheorem{proposition}[proposition]{Proposition}
\aliascntresetthe{proposition}

\newaliascnt{lemma}{theorem}
\newtheorem{lemma}[lemma]{Lemma}
\aliascntresetthe{lemma}

\newaliascnt{corollary}{theorem}
\newtheorem{corollary}[corollary]{Corollary}
\aliascntresetthe{corollary}

\newaliascnt{claim}{theorem}
\newtheorem{claim}[claim]{Claim}
\aliascntresetthe{claim}

\theoremstyle{definition}

\newaliascnt{definition}{theorem}

\aliascntresetthe{definition}

\newaliascnt{example}{theorem}

\aliascntresetthe{example}

\newaliascnt{assumption}{theorem}

\aliascntresetthe{assumption}

\newaliascnt{condition}{theorem}

\aliascntresetthe{condition}

\newaliascnt{question}{theorem}

\aliascntresetthe{question}

\newaliascnt{remark}{theorem}

\aliascntresetthe{remark}

\newaliascnt{remarks}{theorem}

\aliascntresetthe{remarks}

\newaliascnt{aside}{theorem}

\aliascntresetthe{aside}

\newaliascnt{note}{theorem}

\aliascntresetthe{note}

\usepackage{thmtools}
\usepackage{thm-restate}

\usepackage{hyperref}

\hypersetup{
    hypertexnames=false, %
    colorlinks=true,
    linkcolor=Maroon,
    filecolor=Thistle,
    urlcolor=Thistle,
    citecolor=Thistle,
}

\usepackage[nameinlink]{cleveref}

\crefname{theorem}{theorem}{theorems}
\Crefname{theorem}{Theorem}{Theorems}
\crefname{proposition}{proposition}{propositions}
\Crefname{proposition}{Proposition}{Propositions}
\crefname{lemma}{lemma}{lemmas}
\Crefname{lemma}{Lemma}{Lemmas}
\crefname{corollary}{corollary}{corollaries}
\Crefname{corollary}{Corollary}{Corollaries}
\crefname{claim}{claim}{claims}
\Crefname{claim}{Claim}{Claims}
\crefname{definition}{definition}{definitions}
\Crefname{definition}{Definition}{Definitions}
\crefname{example}{example}{examples}
\Crefname{example}{Example}{Examples}
\crefname{assumption}{assumption}{assumptions}
\Crefname{assumption}{Assumption}{Assumptions}

\makeatletter
\let\cref@old@isrefconsecutive\cref@isrefconsecutive
\def\cref@isrefconsecutive#1#2{%
  \begingroup
    \def\cref@assumptiontype{assumption}%
    \cref@gettype{#1}{\cref@typea}%
    \ifx\cref@typea\cref@assumptiontype
      \endgroup
      \@cref@refconsecutivefalse
    \else
      \endgroup
      \cref@old@isrefconsecutive{#1}{#2}%
    \fi
}
\makeatother

\crefname{condition}{condition}{conditions}
\Crefname{condition}{Condition}{Conditions}
\crefname{question}{question}{questions}
\Crefname{question}{Question}{Questions}
\crefname{remark}{remark}{remarks}
\Crefname{remark}{Remark}{Remarks}
\crefname{remarks}{remarks}{remarks}
\Crefname{remarks}{Remarks}{Remarks}
\crefname{aside}{aside}{asides}
\Crefname{aside}{Aside}{Asides}
\crefname{note}{note}{notes}
\Crefname{note}{Note}{Notes}
\crefname{appendix}{appendix}{appendices}
\Crefname{appendix}{Appendix}{Appendices}

\definecolor{backcolour}{rgb}{0.63,0.79,0.95}
\lstdefinestyle{mystyle}{
  backgroundcolor=\color{backcolour},
  basicstyle=\ttfamily\footnotesize,
  breakatwhitespace=false,
  breaklines=true,
  captionpos=b,
  keepspaces=true,
  numbers=left,
  numbersep=5pt,
  showspaces=false,
  showstringspaces=false,
  showtabs=false,
  tabsize=2
}
\AddToHook{cmd/maketitle/before}{\begingroup\everymath{}}
\AddToHook{cmd/maketitle/after}{\endgroup}
\AddToHook{cmd/@makefnmark/before}{\begingroup\everymath{}}
\AddToHook{cmd/@makefnmark/after}{\endgroup}

\begin{document}

\title{Mediated Disclosure}

\author{Mark Whitmeyer\thanks{Arizona State University. Email: \href{mailto:mark.whitmeyer@gmail.com}{mark.whitmeyer@gmail.com}.} \and Kun Zhang\thanks{Auburn University and University of Queensland. Email: \href{mailto:kun@kunzhang.org}{kun@kunzhang.org}.}}
\date{October 27, 2026}
\maketitle

\begin{abstract}
We introduce a mediator to a communication setting with verifiable (hard) evidence and characterize precisely when mediation can induce every obedient outcome--every type-contingent distribution of actions that the receiver is willing to follow--including the outcome that delivers the sender's commitment (Bayesian persuasion) payoff. For both type-independent and general preferences, it is necessary for each type to have a message he alone can send. This is also sufficient for type-independent preferences, whereas the general case also requires a particular alignment of the types' utilities. We also rank evidence structures by the distributions of actions conditional on the sender's type that they can induce at (mediated) equilibrium, through garbling conditions analogous to Blackwell’s comparison of experiments. 
\end{abstract}

\section{Introduction}

Mediation is a powerful tool in communication settings, allowing for greater fine-tuning of reporting and responding incentives, to the potential benefit of the communicating parties. So far, the literature has primarily focused on the benefits of mediation when communication is cheap talk.\footnote{An important exception is \citet{forgeskoessler2005}, who study communication when available messages depend on type. Related work on hard evidence asks when certifiable communication can support full revelation; see \citet{mathis2008,hagenbachkoesslerperezrichet2014}.} Here, we instead specialize to games with verifiable disclosure, \textit{viz.}, hard evidence.

Our starting point is the observation that, with type-independent sender preferences, mediation grants the sender a significant amount of power in the canonical disclosure environments in which every type can prove who he is,\footnote{See, e.g., \citet{grossmanhart1980}, \citet{grossman1981}, \citet{milgrom1981}, and \citet{milgromroberts1986}.} which property of evidence structures we term \textit{complete provability}. In particular, complete provability allows the \textit{informed} sender to generate in equilibrium the Bayesian persuasion outcome: it is as if he could commit \textit{ex ante} to an information structure, reporting incentives be damned. The logic is straightforward. Each type submits a message that only he can send, and the mediator garbles these messages exactly as the sender-optimal experiment garbles the state. Every other message leads to the sender's least-preferred action among those used by the scheme, so no type gains by deviating.

\Cref{thm:complete} proves the converse. With type-independent sender preferences, complete provability is necessary and sufficient for mediation to induce every obedient outcome--every type-contingent distribution of actions that the receiver is willing to follow--and, equivalently, to grant the sender his Bayesian-persuasion value for every receiver payoff. It is necessary even for a weaker conclusion: \textit{whenever persuasion strictly improves on direct disclosure, mediation does too.} If a type lacks an exclusive message, every message he can send is also available to another type. We choose receiver preferences under which mediation can improve only by giving the type without an exclusive message a strictly higher payoff than every other type. Another type who can send the same message can obtain the same lottery, so mediation cannot create this difference in payoffs. Persuasion can, because it conditions directly on type.

Complete provability is stronger than necessary for mediation to help for \textit{some} receiver payoff. \Cref{prop:some} reveals that one informative message--one message that rules out at least one type--is enough. In the proof, we choose receiver preferences so that the sender's preferred action yields her payoff \(\varepsilon>0\) when his type can submit the informative message and \(-1\) otherwise, while his least-preferred action gives her zero. The mediator uses the informative message to support the sender's preferred action and also recommends that action with a small positive probability after other messages. The receiver sees only the recommendation, so she may follow it even though she would choose another action after seeing those messages directly. With two actions, mediation can help only if an informative message exists.

When sender payoffs depend on type, complete provability by itself is insufficient. Every type can send the common message that discloses nothing, so the lottery attached to that message must deter every type from withholding his evidence. \Cref{thm:sd-value} reveals that complete provability and a condition we term \textit{compatible action rankings} are necessary and sufficient to induce every obedient outcome.\footnote{The condition requires that no two types rank a pair of actions strictly in opposite directions. It ensures that every set of actions contains an action that is worst for every type, which the mediator can use after nondisclosure.}

Disagreement over some actions does not eliminate every gain from mediation. \Cref{prop:sd-some-improvement} shows that one informative message is enough if every type strictly prefers the same one of two actions. \Cref{prop:sd-two-exclusive} highlights that a gain is possible if two types have exclusive messages and both strictly prefer the same one of two actions. With two actions and complete provability, \Cref{cor:sd-complete-binary} proves the converse: mediation helps for some receiver payoff exactly when two types strictly prefer the same action.

We finish by comparing evidence structures. For any set \(B\) of types, an evidence structure determines the ways of assigning each type in \(B\) a distribution over messages that he can send and that no type outside \(B\) can send. In our first order, \eqref{eq:evidence-order}, evidence structure \(1\) is at least as powerful as evidence structure \(2\) when every such assignment under the latter can be reproduced by first generating messages under the former and then randomly converting them into messages under the latter. In \Cref{thm:evidence-order}, we find that this condition is equivalent to two comparisons for every set of actions, every full-support prior, and all sender and receiver preferences, with sender preferences independent of type: evidence structure \(1\) can induce every mediated outcome that evidence structure \(2\) can induce, and the sender's best mediated payoff is never lower under evidence structure \(1\).\footnote{The equation in \eqref{eq:evidence-order} has the same form as \citeauthor{blackwell1953}'s \citeyearpar{blackwell1953} comparison: one experiment is less informative when it can be obtained by garbling another. \citet{gossner2000} compares information structures by the equilibrium outcomes they support and characterizes the comparison by whether one structure can reproduce the distributions induced by the other. Here, the comparison applies to all message distributions permitted by an evidence structure rather than to a single experiment.}

Call a set of types \emph{provable} if each type in the set has a message that no type outside the set can send. Provable sets alone need not determine which mediation outcomes an evidence structure can induce: two evidence structures may have the same provable sets but differ in how well their messages distinguish the types within each set. Condition \eqref{eq:evidence-H} rules out this possibility: within every provable set, it requires that the types can be assigned different messages that no type outside the set can send. If evidence structure \(1\) satisfies this condition, \Cref{thm:evidence-H} reveals that an obedient outcome can be induced exactly when, for every real number \(c\), the set of types whose expected payoff is at least \(c\) is provable. There, we also uncover the fact that evidence structure \(1\) can induce every mediated outcome under evidence structure \(2\) exactly when every set provable under the latter is also provable under the former. Without \eqref{eq:evidence-H}, two evidence structures can have the same provable sets even though one yields a strictly higher mediated payoff for some choice of the set of actions and both players' preferences.

We also characterize the order when sender payoffs may depend on type. Matching the distribution of messages is no longer enough. Randomly converting one message into another may give a type a lottery that he could not obtain under the second evidence structure. Consequently, our order in \eqref{eq:evidence-state-dependent} permits that conversion only when every type who can send the first message can also send the second. This condition characterizes when evidence structure \(1\) can induce every mediated outcome under evidence structure \(2\) for sender types with compatible rankings.

\subsection{Related work}

The first part of the paper connects two literatures. One studies what mediation can achieve when messages are unverifiable.\footnote{\citet{salamanca2021} characterizes sender-optimal mediation in cheap-talk games. \citet{corraodai2026} compare mediation with cheap talk and Bayesian persuasion when the sender's payoff does not depend on his type. \citet{kamenicalin2026} prove that, generically, the sender values commitment exactly when he values randomization under commitment and that mediation improves on cheap talk exactly when the sender values the mediator's ability to randomize.} These models generally allow every type to send every message. \citet{forgeskoessler2005} is a notable exception that allows the messages available to an informed player to depend on his type.

The other literature studies communication with hard evidence. It asks when direct communication fully reveals the sender's type and which outcomes can be implemented.\footnote{\citet{lipmanseppi1995} introduce the notion of complete provability. \citet{seidmannwinter1997}, \citet{mathis2008}, and \citet{hagenbachkoesslerperezrichet2014} study full revelation, while \citet{koesslerperezrichet2019} study implementation through mechanisms that condition on submitted evidence. \citet{titovazhang2025} and \citet{shishkintitovazhang2026} compare direct verifiable disclosure with commitment.} We instead ask when adding a mediator makes every obedient outcome attainable.

The second part of the paper compares evidence structures. This connects to work that ranks experiments and information structures.\footnote{\citet{blackwell1953} compares experiments by garbling. \citet{gossner2000} compares information structures by the equilibrium outcomes they support. \citet{gentzkowkamenica2017} identify a condition on the signals available to competing persuaders under which competition reveals at least as much information as collusion for every preference profile.} A single evidence structure allows the sender to use many different strategies for choosing messages, subject to the messages available to each type. We characterize when every mediated outcome under one evidence structure can also be induced under another.

\section{Model}\label{sec:model}

A sender (he) privately observes his type \(\theta\in\Theta\), drawn from a full-support prior \(\mu\) on a finite set \(\Theta\). The receiver (she) chooses an action from a finite set \(A\). Her payoff is \(u(a,\theta)\), while the sender's payoff is \(v(a,\theta)\), where \(u,v\colon A\times\Theta\to\R\). We assume that each sender type strictly prefers some action to another: \(v(\cdot,\theta)\) is nonconstant for every \(\theta\in\Theta\). Where explicitly indicated, we also allow types who are indifferent among all actions.

The sender's message space \(M\) is finite, and \(M(\theta)\subseteq M\) is the nonempty set of messages available to type \(\theta\).
Together, \(M\) and the family \(\{M(\theta)\}_{\theta\in\Theta}\) define the \emph{evidence structure} \(\mathcal E\).
For each message \(m\), let
\[E_m\coloneqq\{\theta\in\Theta\colon m\in M(\theta)\}\]
denote the set of types who can submit it.
We discard messages available to no type. Every type has access to a common message \(m_0\), interpreted as ``disclose nothing,'' i.e., \(m_0\in M(\theta)\) for every \(\theta\in\Theta\). A message \(m\) is \emph{informative} if \(E_m\ne\Theta\), meaning that observing it rules out at least one type. It is \emph{exclusive to type \(\theta\)} if \(E_m=\{\theta\}\). Following \citet{lipmanseppi1995}, the evidence structure has \emph{complete provability} if every type has an exclusive message. \citet{hagenbachkoesslerperezrichet2014} call the same property \emph{own type certifiability}.

We consider two communication settings. Under \emph{direct disclosure}, the sender first privately observes his type \(\theta\), then submits a message \(m\in M(\theta)\) directly to the receiver. After observing \(m\), the receiver chooses an action \(a\in A\).

Under \emph{mediation}, a mediator first chooses a finite signal space \(Y\) and publicly commits to a stochastic map \(g\colon M\to\Delta(Y)\),\footnote{For any finite set \(X\), \(\Delta(X)\) denotes the set of probability distributions on \(X\), and, for each \(x\in X\), \(\delta_x\) denotes the distribution assigning probability one to \(x\).} before the sender learns his type. The sender then privately observes \(\theta\) and submits a message \(m\in M(\theta)\) to the mediator. After observing \(m\), the mediator draws a signal \(y\in Y\) according to \(g(\cdot\mid m)\). The receiver observes \(y\), but not \(m\), and chooses an action \(a\in A\).

In either setting, a sender strategy assigns to each type \(\theta\) a probability distribution \(\sigma(\cdot\mid\theta)\) with support on \(M(\theta)\). Under mediation, a receiver strategy \(r\colon Y\to\Delta(A)\) assigns probability \(r(a\mid y)\) to action \(a\) after signal \(y\); under direct disclosure, \(r\colon M\to\Delta(A)\) assigns an action distribution to each observed message.

In both settings, we use \emph{sequential equilibrium} as defined by \citet{krepswilson1982}. An equilibrium consists of strategies and beliefs satisfying sequential rationality and consistency. Consistency requires a single sequence of completely mixed strategy profiles converging to the equilibrium strategies, with the associated Bayesian beliefs converging to the equilibrium beliefs.

A \emph{mediated equilibrium} consists of a choice of \((Y,g)\) together with a sequential equilibrium of the resulting game between the sender and receiver, who take \((Y,g)\) as given. Under direct disclosure, consistency implies that the receiver's beliefs assign probability one to \(E_m\) after every message \(m\), including off-path messages.

Direct disclosure and mediation differ in the messages and signals observed by the receiver. To compare what they achieve, we consider, for each sender type, the probability that the receiver chooses each action. These probabilities determine both players' expected payoffs and allow us to compare communication settings even when they convey information differently. An \emph{outcome} is therefore a map \(q\colon\Theta\to\Delta(A)\), where \(q(a\mid\theta)\) is the probability that the receiver chooses action \(a\) when the sender has type \(\theta\). An equilibrium \emph{induces} an outcome \(q\) if, conditional on each sender type \(\theta\), the receiver's action in that equilibrium has distribution \(q(\theta)\). An outcome is a \emph{mediation outcome} if it is induced by a mediated equilibrium.

Receiver optimality restricts which outcomes can arise in equilibrium: she cannot strictly increase her expected payoff by choosing \(b\) whenever she would otherwise choose \(a\). We call an outcome \(q\) \emph{obedient} if it satisfies this requirement for every pair of actions:
\begin{equation}\label{eq:obedience}
\sum_{\theta\in\Theta}\mu(\theta)q(a\mid\theta)\left[u(a,\theta)-u(b,\theta)\right]\ge0\qquad\text{for every }a,b\in A.
\end{equation}
Let \(\mathcal Q(\mu,u)\) denote the set of obedient outcomes. Every such outcome can be induced by a committed signal rule whose distribution depends directly on the sender's type.\footnote{Use the actions that occur with positive probability under \(q\) as signals, and draw signal \(a\) with probability \(q(a\mid\theta)\) at type \(\theta\). After observing signal \(a\), the receiver chooses action \(a\), which is optimal by obedience.} This provides a benchmark for assessing the additional restrictions imposed by evidence availability and the sender's incentives to submit it.

In particular, if the sender could choose this rule before learning his type, he would select an obedient outcome maximizing his \textit{ex ante} expected payoff. This is the Bayesian persuasion problem of \citet{kamenicagentzkow2011}. We denote its value by \(V^{\mathrm{BP}}(\mu,u,v)\):
\begin{equation}\label{eq:bp}
V^{\mathrm{BP}}(\mu,u,v)\coloneqq\max_{q\in\mathcal Q(\mu,u)}\sum_{\theta\in\Theta}\mu(\theta)\sum_{a\in A}v(a,\theta)q(a\mid\theta).
\end{equation}
Let \(V^{\mathrm D}(\mu,u,v)\) denote the sender's highest equilibrium payoff under direct disclosure, and let \(V^{\mathrm{MD}}(\mu,u,v)\) be the supremum of his payoffs over all mediated equilibria, allowing both \((Y,g)\) and the associated equilibrium to vary. Since the mediator can be transparent and every equilibrium induces an obedient outcome, \(V^{\mathrm D}(\mu,u,v)\le V^{\mathrm{MD}}(\mu,u,v)\le V^{\mathrm{BP}}(\mu,u,v)\).

\section{Evidence and the gains from mediation}\label{sec:free}

When can mediation attain the sender's information-design maximum for every receiver payoff function, and when can it strictly improve on direct disclosure for some receiver payoff function? We study these questions holding the evidence structure, the action set \(A\), and the sender's payoff function \(v\) fixed, while varying the receiver's payoff function \(u\). The prior \(\mu\) is fixed unless stated otherwise. We first consider type-independent sender preferences and then allow the sender's payoff to depend on his type. All proofs are relegated to the appendix.

\subsection{Type-independent sender preferences}\label{subsec:state-independent}

Suppose the sender's payoff does not depend on his type, and write \(v(a,\theta)=v(a)\). We first characterize the evidence structures for which every obedient outcome is a mediation outcome for every receiver payoff.

\begin{theorem}\label{thm:complete}
The following statements are equivalent:
\begin{enumerate}[label=\textup{(\roman*)},nosep]
\item The evidence structure has complete provability.
\item For every receiver payoff \(u\), every \(q\in\mathcal Q(\mu,u)\) is a mediation outcome.
\item For every receiver payoff \(u\), \(V^{\mathrm{MD}}(\mu,u,v)=V^{\mathrm{BP}}(\mu,u,v)\).
\item For every receiver payoff \(u\),
\begingroup
\setlength{\belowdisplayskip}{0pt}
\setlength{\belowdisplayshortskip}{0pt}
\[V^{\mathrm{BP}}(\mu,u,v)>V^{\mathrm D}(\mu,u,v)\quad\Longrightarrow\quad V^{\mathrm{MD}}(\mu,u,v)>V^{\mathrm D}(\mu,u,v).\]
\endgroup \label{complete-iv}
\end{enumerate}
Under complete provability, every obedient outcome can be induced using a pure sender strategy and at most \(|A|\) signals.
\end{theorem}

Complete provability enables the mediator to verify the sender's type and then choose how much of that information reaches the receiver. To induce an obedient outcome \(q\), each type submits his exclusive message and the mediator recommends an action according to \(q(\theta)\). Every other message leads to the sender's least-preferred action among those used by \(q\), making the prescribed submission optimal. The receiver is willing to follow the recommendations because \(q\) is obedient.

Statement \ref{complete-iv} asks only whether mediation can strictly increase the sender's payoff above his best direct-disclosure payoff whenever information design can. To see why complete provability is necessary for this statement to hold, suppose a type \(\theta^*\) lacks an exclusive message, and take any message used by him. Sender optimality implies that this message gives him his equilibrium expected payoff. Since it is also available to another type, and all sender types have the same preferences, this other type can obtain the same payoff by submitting the message and, accordingly, must receive at least as much in equilibrium.

In the proof we construct quadratic-loss receiver payoffs for which only the sender's most- and least-preferred actions can be optimal for her. Direct disclosure gives the sender only the payoff from his least-preferred action, while the receiver is willing to choose his most-preferred action only when she assigns sufficiently high probability to \(\theta^*\). By Bayes' rule, this requires the probability of choosing such an action to be strictly higher when the sender's type is \(\theta^*\) than when it is any other type. Thus, for mediation to give the sender strictly more than direct disclosure, his expected payoff would have to be strictly higher at \(\theta^*\) than at every other type. But every message used by \(\theta^*\) is available to some other type, and that type can obtain the same lottery and must, therefore, receive at least as much in equilibrium. Under information design, in contrast, recommendations depend directly on the sender's type, so no other type can mimic \(\theta^*\) to obtain his payoff. The receiver can thus choose the sender's most-preferred action at \(\theta^*\) and his least-preferred action elsewhere, giving him strictly more than under direct disclosure. Consequently, some receiver payoff \(u\) satisfies \(V^{\mathrm{MD}}(\mu,u,v)=V^{\mathrm D}(\mu,u,v)<V^{\mathrm{BP}}(\mu,u,v)\).

We now ask a less demanding question: what evidence suffices for mediation to strictly improve on direct disclosure for at least \emph{some} receiver preferences?
\begin{proposition}\label{prop:some}
If an informative message is available, then there exists a receiver payoff \(u\) such that \(V^{\mathrm{MD}}(\mu,u,v)>V^{\mathrm D}(\mu,u,v)\). When \(|A|=2\), the existence of an informative message is also necessary.
\end{proposition}

In the proof, we choose receiver payoffs such that, if the receiver knew the sender's type, she would strictly prefer the sender's most-preferred action when he can submit the informative message, and his least-preferred action otherwise. The receiver's gain from choosing the sender's most-preferred action when the informative message is available is sufficiently small that, under direct disclosure, this action can be induced only when the sender can submit that message. Under mediation, the mediator can induce the sender's most-preferred action whenever he submits the informative message, while also recommending it with a small positive probability after other messages. Because the receiver observes only the recommendation, not the submitted message, she is willing to follow it as long as that probability is sufficiently small. 
And so mediation yields the sender a strictly higher payoff than his best direct-disclosure equilibrium.

\subsection{General sender preferences}\label{subsec:state-dependent-complete}

We now allow the sender's payoff to depend on his type. The question is whether complete provability can still ensure, for every receiver payoff, that every obedient outcome is a mediation outcome when sender types disagree over the rankings of actions. We say that sender types have \emph{compatible action rankings} if
\begin{equation}\label{eq:sd-compatible-rankings}
\left[v(a,\theta)-v(b,\theta)\right]\left[v(a,\widehat\theta)-v(b,\widehat\theta)\right]\ge0\qquad\text{for every }a,b\in A\text{ and }\theta,\widehat\theta\in\Theta.
\end{equation}
Thus, two types never rank a pair of actions strictly in opposite directions, although one type may be indifferent between actions that another type ranks strictly.

When the sender's payoff depends on his type, two actions can give him the same \textit{ex ante} payoff even though he strictly prefers one action at some types and the other at other types. We consider the following nondegeneracy condition: if two actions give the sender the same \textit{ex ante} payoff, they also give him the same payoff at every type. Formally, for every pair of actions \(a,b\in A\):
\begin{equation}\label{eq:sd-nondegeneracy}
\sum_{\theta\in\Theta}\mu(\theta)v(a,\theta)=\sum_{\theta\in\Theta}\mu(\theta)v(b,\theta)\quad\Longrightarrow\quad v(a,\theta)=v(b,\theta)\quad\text{for every }\theta\in\Theta.
\end{equation}

\begin{theorem}\label{thm:sd-value}
Suppose \eqref{eq:sd-nondegeneracy} holds. The following statements are equivalent:
\begin{enumerate}[label=\textup{(\roman*)},nosep]
\item The evidence structure has complete provability, and sender types have compatible action rankings. \label{sd-complete}
\item For every receiver payoff \(u\), every \(q\in\mathcal Q(\mu,u)\) is a mediation outcome. \label{sd-allob}
\item For every receiver payoff \(u\), \(V^{\mathrm{MD}}(\mu,u,v)=V^{\mathrm{BP}}(\mu,u,v)\). \label{sd-bp}
\item For every receiver payoff \(u\),
\begingroup
\setlength{\belowdisplayskip}{0pt}
\setlength{\belowdisplayshortskip}{0pt}
\[V^{\mathrm{BP}}(\mu,u,v)>V^{\mathrm D}(\mu,u,v)\quad\Longrightarrow\quad V^{\mathrm{MD}}(\mu,u,v)>V^{\mathrm D}(\mu,u,v).\]
\endgroup \label{sd-improve-all}
\end{enumerate}
The equivalence of statements \ref{sd-complete} and \ref{sd-allob} does not require \eqref{eq:sd-nondegeneracy}. Whenever statement \ref{sd-complete} holds, every obedient outcome can be induced using a pure sender strategy and at most \(|A|\) signals.
\end{theorem}

Compatible action rankings ensure that, among the actions used by any obedient outcome, there is one that minimizes every sender type's payoff. Consequently, the construction described after \Cref{thm:complete} still makes every type willing to submit his exclusive message.

To see why compatible action rankings are necessary, suppose two types rank \(a\) and \(b\) strictly in opposite directions. Choose receiver preferences such that, for every sender type, she is indifferent between these actions and strictly prefers them to every other action. Consider an outcome that assigns only actions \(a\) and \(b\) and is thus obedient: the outcome assigns \(b\) to the type who strictly prefers \(a\) to \(b\), and \(a\) to the type who strictly prefers \(b\) to \(a\). However, both types can submit \(m_0\) and obtain the same distribution over \(a\) and \(b\), so at least one would gain by withholding evidence. As a result, for these receiver preferences, the obedient outcome just described cannot be induced by any mediated equilibrium, even under complete provability.

Under the nondegeneracy condition \eqref{eq:sd-nondegeneracy}, if the evidence structure lacks complete provability or two sender types rank some pair of actions strictly in opposite directions, there exists a receiver payoff \(u\) such that \(V^{\mathrm{MD}}(\mu,u,v)=V^{\mathrm D}(\mu,u,v)<V^{\mathrm{BP}}(\mu,u,v)\).\footnote{In the Supplementary Appendix we provide a weaker nondegeneracy condition under which the conclusions of \Cref{thm:sd-value} still hold as well as an example in which mediation attains the information-design maximum for every receiver payoff function despite the absence of complete provability and compatible action rankings.} Accordingly, information design can strictly improve on the sender's best direct-disclosure payoff, while mediation cannot.

When types rank a pair of actions oppositely, their prior-weighted payoff differences may cancel at a particular prior, but cannot cancel at every full-support prior. Requiring \Cref{sd-allob,sd-bp,sd-improve-all} to hold at every full-support prior, therefore, removes the need for \eqref{eq:sd-nondegeneracy}.

\begin{corollary}\label{cor:sd-all-priors}
Without \eqref{eq:sd-nondegeneracy}, the four statements in \Cref{thm:sd-value} are equivalent when the requirements in \Cref{sd-allob,sd-bp,sd-improve-all} are imposed at every full-support prior.
\end{corollary}

We next consider when mediation can strictly improve on direct disclosure for some receiver payoff function.

\begin{proposition}\label{prop:sd-some-improvement}
\begin{enumerate}[label=\textup{(\roman*)},nosep]
\item\label{sd-some-informative} Suppose there are actions \(a^+,a^-\) such that \(v(a^+,\theta)>v(a^-,\theta)\) for every \(\theta\in\Theta\). If an informative message is available, there is a receiver payoff \(u\) such that \(V^{\mathrm{MD}}(\mu,u,v)>V^{\mathrm D}(\mu,u,v)\). A mediated equilibrium with a pure sender strategy and two signals can strictly improve on direct disclosure.
\item\label{sd-some-binary} If \(|A|=2\) and no informative message is available, then \(V^{\mathrm{MD}}(\mu,u,v)=V^{\mathrm D}(\mu,u,v)\) for every receiver payoff \(u\) and every sender payoff \(v\), including payoffs for which some types are indifferent among all actions.
\end{enumerate}
\end{proposition}

Part \ref{sd-some-informative} extends the construction described after \Cref{prop:some}. In the proof, we choose receiver payoffs for which only \(a^+\) and \(a^-\) can be optimal. Every type's payoff is strictly increasing in the probability of \(a^+\), so the sender strategy in that construction remains optimal. Compared with the sender's best direct-disclosure equilibrium, mediation leaves the payoffs of types who can submit the informative message unchanged. It strictly benefits every other type by giving him a positive probability of \(a^+\) instead of \(a^-\) with certainty.

Part \ref{sd-some-binary} follows by showing that any mediated equilibrium can be replaced by a direct-disclosure equilibrium that makes no sender type worse off. For intuition, suppose every sender type strictly prefers one of the two actions. Because every type can submit every message, each type chooses a message that maximizes the probability of his preferred action. If every message induces the same lottery over the two actions, that lottery is a receiver best response under the prior and can be reproduced by direct disclosure. If different messages induce different lotteries, the proof constructs an obedient outcome that gives every type his preferred action with certainty. Direct disclosure can implement this outcome by having each type choose a message inducing his preferred action.

Part \ref{sd-some-informative} of \Cref{prop:sd-some-improvement} requires all sender types to agree on a strict comparison between two actions, but only one informative message is needed. We next show that agreement between just two types suffices if each of the two types has an exclusive message, even when other types have opposite preferences.

\begin{proposition}\label{prop:sd-two-exclusive}
Suppose distinct types \(\theta_1,\theta_2\) each have an exclusive message, and there are actions \(a,b\) such that for each \(i=1,2\), \(v(a,\theta_i)>v(b,\theta_i)\). Then there is a receiver payoff \(u\) for which \(V^{\mathrm{MD}}(\mu,u,v)>V^{\mathrm D}(\mu,u,v)\). A mediated equilibrium with a pure sender strategy and two signals can strictly improve on direct disclosure.
\end{proposition}

Our proof adapts the construction described after \Cref{prop:some}. Receiver payoffs are chosen so that the sender's best direct-disclosure payoff is attained when \(\theta_1\) submits his exclusive message and all other types submit \(m_0\). The receiver chooses \(a\) after \(\theta_1\)'s exclusive message and \(b\) after \(m_0\). Under mediation, the mediator recommends \(a\) with certainty after \(\theta_1\)'s exclusive message and with a small positive probability after \(\theta_2\)'s exclusive message. All other messages induce \(b\).

Both \(\theta_1\) and \(\theta_2\) optimally submit their exclusive messages: neither can submit the other's, and both strictly prefer \(a\) to \(b\). Every other type can induce only \(b\), because neither exclusive message is available to him. As in \Cref{prop:some}, the receiver is willing to follow both recommendations when the probability of recommending \(a\) after \(\theta_2\)'s message is sufficiently small. Mediation thus improves upon the sender's best direct-disclosure equilibrium, strictly increasing \(\theta_2\)'s payoff and leaving every other type's payoff unchanged.\footnote{We provide a weaker condition on messages in the supplementary appendix that does not require either message to be exclusive.}

With two actions and complete provability, the condition in \Cref{prop:sd-two-exclusive} is also necessary for mediation to strictly improve on direct disclosure for some receiver payoff.

\begin{corollary}\label{cor:sd-complete-binary}
Suppose \(|A|=2\) and complete provability holds. There exists a receiver payoff \(u\) for which \(V^{\mathrm{MD}}(\mu,u,v)>V^{\mathrm D}(\mu,u,v)\) if and only if two distinct types strictly prefer the same action.

In particular, with at least three types, mediation can strictly improve on direct disclosure for some receiver preferences even when sender types have conflicting preferences.
\end{corollary}

With only two actions, our assumption that \(v(\cdot,\theta)\) is nonconstant for every type \(\theta\) means that each type strictly prefers one action to the other. Among three or more types, at least two must therefore prefer the same action. Complete provability gives these types exclusive messages, so \Cref{prop:sd-two-exclusive} implies that, for some receiver preferences, mediation gives the sender a strictly higher expected payoff than he can obtain in any direct-disclosure equilibrium.

By contrast, with two types who strictly prefer opposite actions, mediation cannot give the sender a higher expected payoff than his best direct-disclosure equilibrium for any receiver preferences. This holds even though complete provability gives each type a message that rules out the other type and thus identifies him. If the nondegeneracy condition \eqref{eq:sd-nondegeneracy} also holds, \Cref{thm:sd-value} implies that, for some receiver preferences, information design gives the sender a strictly higher expected payoff than his best direct-disclosure equilibrium, while mediation cannot.

\section{Ranking evidence structures}\label{sec:evidence-comparison}

So far, we have explored the gains from mediation for a given evidence structure. We now compare different evidence structures, asking when every mediation outcome under one is also a mediation outcome under another, and how this relates to comparisons of the sender's highest mediated payoff. As before, we first consider type-independent sender preferences and then turn to the general case.

Let \(\mathcal E^1\) and \(\mathcal E^2\) be evidence structures on the same type space \(\Theta\). We add a superscript \(i\) to the notation \(M\), \(M(\theta)\), and \(E_m\) to distinguish the two structures, and write \(V_i^{\mathrm{MD}}(\mu,u,v)\) for \(V^{\mathrm{MD}}(\mu,u,v)\) under \(\mathcal E^i\).

\subsection{Type-independent sender preferences}\label{subsec:evidence-state-independent}

Suppose the sender's payoff does not depend on his type, and write \(v(a,\theta)=v(a)\).

We begin by defining \emph{provable sets of types} (\emph{provable sets} for short). A set \(B\subseteq\Theta\) is provable under evidence structure \(i\) if every type in \(B\) has an available message that rules out every type outside \(B\); that is,
\[
\text{for every }\theta\in B,\quad \text{there exists }m\in M^i\text{ such that }\theta\in E_m^i\subseteq B.
\]
In particular, complete provability is equivalent to every set of types being provable.

Provability requires messages that rule out types outside a set. We also need to consider what these messages reveal about types within the set. For each nonempty \(B\subseteq\Theta\), let
\[
M_B^i\coloneqq\{m\in M^i\colon E_m^i\subseteq B\}
\]
be the set of messages that rule out every type outside \(B\). Each type in \(B\) can randomize over the messages in \(M_B^i\) that are available to him, and we represent such a sender strategy by a \emph{reporting matrix} \(S\), where \(S_{\theta m}=\sigma(m\mid\theta)\) is the probability that type \(\theta\) submits message \(m\). The set of all such matrices is
\[
\Sigma_i(B)\coloneqq\left\{S\in\R_+^{B\times M_B^i}\colon\sum_{m\in M_B^i}S_{\theta m}=1\text{ for every }\theta\in B,\quad S_{\theta m}=0\text{ if }\theta\notin E_m^i\right\}.
\]
Each row is a probability distribution that assigns zero probability to unavailable messages. The set \(\Sigma_i(B)\) is nonempty exactly when \(B\) is provable under evidence structure \(i\).

We write \(\mathcal E^1\succeq\mathcal E^2\) if, for every nonempty set of types \(B\), every reporting matrix under evidence structure \(\mathcal E^2\) can be obtained by garbling a reporting matrix under evidence structure \(\mathcal E^1\), with both using messages that rule out every type outside \(B\). Formally,
\begin{equation}\tag{E}\label{eq:evidence-order}
\begin{gathered}\text{for every nonempty }B\subseteq\Theta\text{ and every }S^2\in\Sigma_2(B),\\ \text{there exist }S^1\in\Sigma_1(B)\text{ and a row-stochastic matrix }K\\ \text{such that }S^2=S^1K.\end{gathered}
\end{equation}
Here \(K_{mm'}\) is the probability of converting message \(m\) under evidence structure \(\mathcal E^1\) into message \(m'\) under evidence structure \(\mathcal E^2\). The same conversion rule applies to every type. The equality \(S^2=S^1K\) requires the resulting distribution of messages to equal the corresponding row of \(S^2\) for each type in \(B\).

\begin{theorem}\label{thm:evidence-order}
The following statements are equivalent:
\begin{enumerate}[label=\textup{(\roman*)},nosep]
\item\label{evidence-order-condition} \(\mathcal E^1\succeq\mathcal E^2\).
\item\label{evidence-order-payoff} For every set of actions \(A\), every type-independent sender payoff \(v\colon A\to\R\), every receiver payoff \(u\), and every full-support prior \(\mu\),
\(V_1^{\mathrm{MD}}(\mu,u,v)\ge V_2^{\mathrm{MD}}(\mu,u,v)\).
\item\label{evidence-order-outcomes} For every such \(A,v,u,\mu\), every mediation outcome under evidence structure \(\mathcal E^2\) is also a mediation outcome under evidence structure \(\mathcal E^1\).
\end{enumerate}
\end{theorem}

\Cref{thm:evidence-order} shows that evidence structure \(\mathcal E^1\) reproduces every mediation outcome under evidence structure \(\mathcal E^2\) exactly when it can match every reporting matrix under evidence structure \(\mathcal E^2\) by garbling messages that rule out every type outside the same set.\footnote{The relation \(\succeq\) is reflexive and transitive, and becomes a partial order after identifying evidence structures that dominate each other.} The same condition characterizes when the sender's highest equilibrium payoff under mediation is never lower under evidence structure \(\mathcal E^1\) than under evidence structure \(\mathcal E^2\), regardless of the set of actions \(A\) and both players' preferences. The equivalence in \Cref{thm:evidence-order} also holds when the prior is fixed at any full-support distribution.

The order in \eqref{eq:evidence-order} follows the logic of Blackwell's comparison of experiments. \citet{blackwell1953} asks whether one experiment can be obtained by garbling another, while \citet{gossner2000} applies a related comparison to the equilibrium outcomes supported by information structures. Here, each evidence structure permits a set of reporting matrices, and evidence structure \(\mathcal E^1\) dominates evidence structure \(\mathcal E^2\) when every matrix permitted by \(2\) can be obtained by garbling some matrix permitted by \(1\).

To see why \(\mathcal E^1\succeq\mathcal E^2\) lets evidence structure \(\mathcal E^1\) reproduce every mediation outcome under evidence structure \(\mathcal E^2\), fix a mediated equilibrium under evidence structure \(\mathcal E^2\) and let \(q\) be its outcome. Let \(w(\theta)\coloneqq\sum_{a\in A}v(a)q(a\mid\theta)\) denote the sender's expected payoff conditional on type \(\theta\) under \(q\). For each \(c\in\R\), let \(B_c\coloneqq\{\theta\in\Theta\colon w(\theta)\ge c\}\) be the set of types whose expected payoff is at least \(c\). In this equilibrium, every message used by a type in \(B_c\) delivers at least \(c\). Since the sender's payoff does not depend on his type, a type outside \(B_c\) who could submit a message used by a type in \(B_c\) would also obtain at least \(c\), exceeding his equilibrium payoff. Messages used by types in \(B_c\) must therefore be unavailable to types outside \(B_c\). For nonempty \(B_c\), the equilibrium reporting probabilities of types in \(B_c\) therefore form a matrix in \(\Sigma_2(B_c)\); in particular, \(B_c\) is provable. Condition \eqref{eq:evidence-order} guarantees that this matrix can be reproduced by garbling a reporting matrix under evidence structure \(\mathcal E^1\). The messages used under evidence structure \(\mathcal E^1\) also rule out every type outside \(B_c\).

Provability of these sets is not enough: the messages must also provide enough information about which type in \(B_c\) the sender is to support the receiver's responses. Condition \eqref{eq:evidence-order} preserves both requirements. Applying it at each attained payoff \(c\), the mediator can reproduce the original reporting matrix on \(B_c\) by garbling messages under evidence structure \(\mathcal E^1\) that rule out every type outside \(B_c\). Using this exclusion, we choose a sender strategy in which types with different values of \(w(\theta)\) never use the same message. The mediator can therefore specify one action lottery after each message. After garbling, the mediator uses the action lottery induced by the resulting message in the mediated equilibrium under evidence structure \(\mathcal E^2\) that induces \(q\), reproducing \(q(\theta)\) for every type \(\theta\). A message used under this strategy gives an expected payoff \(c\) and excludes all types with \(w(\theta)<c\). Any type who can submit it therefore already receives at least \(c\), so cannot gain by switching to it. Messages unused by every type instead induce the sender's least-preferred action among those occurring under \(q\), which also cannot improve any type's payoff. Because \(q\) is preserved, the receiver remains willing to follow the recommendations.

If condition \eqref{eq:evidence-order} fails, we can choose the set of actions and both players' preferences so that the sender's highest equilibrium payoff under mediation is strictly higher under evidence structure \(\mathcal E^2\) than under evidence structure \(\mathcal E^1\). We choose receiver preferences that make the sender's preferred actions optimal only at the posteriors generated by a selected reporting matrix on a set \(B\). Under evidence structure \(\mathcal E^2\), a mediated equilibrium uses this matrix and gives types in \(B\) only the sender's most-preferred actions and all other types only his least-preferred actions. We choose preferences so that matching this expected payoff under evidence structure \(\mathcal E^1\) would require the selected reporting matrix to be a garbling of a matrix in \(\Sigma_1(B)\). When \eqref{eq:evidence-order} fails, we can choose the reporting matrix so that no such garbling exists.

We next give a condition on \(\mathcal E^1\) under which provable sets characterize which obedient outcomes are mediation outcomes. This characterization also reduces the comparison in \eqref{eq:evidence-order} to inclusion of provable sets. Let
\[
\mathcal P_i\coloneqq\left\{\bigcup_{m\in D}E_m^i\colon D\subseteq M^i\right\}
\]
be the family of provable sets under evidence structure \(i\). For a single evidence structure, omit the index. Consider the condition
\begin{equation}\tag{H}\label{eq:evidence-H}
 \begin{gathered}
 \text{for every nonempty }B\in\mathcal P,
 \text{ there is an injective map }f_B\colon B\to M\\
 \text{such that }\theta\in E_{f_B(\theta)}\subseteq B
 \text{ for every }\theta\in B.
 \end{gathered}
\end{equation}
In words, \eqref{eq:evidence-H} states that whenever \(B\) is provable, we can assign each type in \(B\) a distinct message, with every assigned message ruling out all types outside \(B\). When types follow this assignment, the mediator can tell which type sent each message. After receiving the message assigned to type \(\theta\), the mediator can randomly generate a report according to \(\theta\)'s row of any desired reporting matrix on \(B\). Thus, if \(\mathcal E^1\) satisfies \eqref{eq:evidence-H}, condition \eqref{eq:evidence-order} holds whenever every set provable under \(\mathcal E^2\) is also provable under \(\mathcal E^1\).

Condition \eqref{eq:evidence-H} is related to the evidence-base concept of \citet{hagenbachkoesslerperezrichet2014}. An evidence base assigns each type a distinct message such that no other type who can submit that message wants to masquerade as that type, relative to the selected continuation equilibrium. Condition \eqref{eq:evidence-H} instead depends only on the evidence structure: within every provable set, it requires distinct messages that rule out every type outside the set.

\begin{theorem}\label{thm:evidence-H}
Suppose evidence structure \(\mathcal E^1\) satisfies \eqref{eq:evidence-H}. Fix any set of actions \(A\), any type-independent sender payoff \(v\), and any full-support prior \(\mu\). The following statements are equivalent:
\begin{enumerate}[label=\textup{(\roman*)},nosep]
\item\label{evidence-H-inclusion} \(\mathcal P_2\subseteq\mathcal P_1\).
\item\label{evidence-H-payoff} For every receiver payoff \(u\),
\(V_1^{\mathrm{MD}}(\mu,u,v)\ge V_2^{\mathrm{MD}}(\mu,u,v)\).
\item\label{evidence-H-outcomes} For every receiver payoff \(u\), every mediation outcome under evidence structure \(\mathcal E^2\) is also a mediation outcome under evidence structure \(\mathcal E^1\).
\end{enumerate}
Moreover, an obedient outcome \(q\) is a mediation outcome under evidence structure \(\mathcal E^1\) if and only if
\begin{equation}\label{eq:evidence-upper-sets}
 \{\theta\in\Theta \colon w(\theta)\ge c\}\in\mathcal P_1
 \qquad\text{for every }c\in\R.
\end{equation}
Whenever this condition holds, a pure sender strategy and at most \(|A|\) signals suffice.
\end{theorem}

\Cref{thm:evidence-H} shows that, under \eqref{eq:evidence-H}, an obedient outcome is a mediation outcome exactly when, for every real number \(c\), the set of types whose expected payoff is at least \(c\) is provable. In particular, two obedient outcomes that give the sender the same expected payoff conditional on each type are either both mediation outcomes or neither is, even if they assign different action lotteries.

We now show how \eqref{eq:evidence-H} lets us induce any obedient outcome \(q\) satisfying \eqref{eq:evidence-upper-sets}. For each type \(\theta\), select the message assigned to him by \eqref{eq:evidence-H} for the set \(B_{w(\theta)}\) and attach the lottery \(q(\theta)\). Types \(\theta\) with \(w(\theta)=c\) use the same assignment for \(B_c\), which gives them distinct messages by \eqref{eq:evidence-H}. Distinct messages allow these types to receive different action lotteries that give the sender the same expected payoff. Because the message assigned to \(\theta\) rules out types whose expected payoff under \(q\) is below \(w(\theta)\), the reasoning used after \Cref{thm:evidence-order} again ensures that types with different values of \(w(\theta)\) receive different assigned messages and that no type gains by switching to another type's assigned message. Every message not assigned to a type induces the sender's least-preferred action among those occurring under \(q\), so no type gains by submitting it. Thus, evidence imposes no further restriction beyond provability of the sets \(B_c\) and receiver obedience. In particular, complete provability satisfies \eqref{eq:evidence-H} and makes this restriction automatic, recovering \Cref{thm:complete}.

Condition \eqref{eq:evidence-H} is imposed only on evidence structure \(\mathcal E^1\); \(\mathcal E^2\) can be any evidence structure. Under this condition, \[\mathcal E^1\succeq\mathcal E^2\quad\Longleftrightarrow\quad\mathcal P_2\subseteq\mathcal P_1.\] Equivalently, for every message \(n\) under \(\mathcal E^2\), each type who can submit \(n\) has an available message under \(\mathcal E^1\) that rules out every type ruled out by \(n\). Unlike \Cref{thm:evidence-order}, this characterization also allows the action set and sender payoff to be fixed; only the receiver's payoff needs to vary.

If evidence structure \(\mathcal E^1\) fails \eqref{eq:evidence-H}, we can construct another evidence structure \(\mathcal E^2\) with exactly the same provable sets, \(\mathcal P_2=\mathcal P_1\). For some choice of the action set and both players' preferences, however, the sender's highest mediated payoff is strictly higher under \(\mathcal E^2\): \(V_2^{\mathrm{MD}}(\mu,u,v)>V_1^{\mathrm{MD}}(\mu,u,v)\). We obtain \(\mathcal E^2\) by adding distinct messages available precisely to the types in an already provable set. These messages allow the mediator to assign different actions to types within that set without making any additional set provable. %

\subsection{General sender preferences}\label{subsec:evidence-state-dependent}

We now allow the sender's payoff to depend on type. The argument for \Cref{thm:evidence-order} uses the fact that a given distribution of actions gives every type the same expected payoff. With type-dependent payoffs, condition \eqref{eq:evidence-order} alone need not preserve the sender's incentives. We consider a condition requiring that, whenever one message is converted into another, every type who can submit the first message can also submit the second.

Write \(\mathcal E^1\succeq^{\mathrm{G}}\mathcal E^2\) if
\begin{equation}\tag{G}\label{eq:evidence-state-dependent}
\begin{gathered}
\text{for every }S^2\in\Sigma_2(\Theta),\\
\text{there exist }S^1\in\Sigma_1(\Theta)\text{ and a row-stochastic matrix }K\\
\text{such that }S^2=S^1K,\\
\text{\&}\qquad K_{mn}>0\Longrightarrow E_m^1\subseteq E_n^2.
\end{gathered}
\end{equation}
Here \(K\) has rows indexed by \(M^1\) and columns indexed by \(M^2\). Condition \eqref{eq:evidence-order} compares reporting matrices separately on each set \(B\), using messages that rule out types outside \(B\), whereas \eqref{eq:evidence-state-dependent} compares reporting matrices on \(\Theta\) and restricts each conversion of messages.

To see why \eqref{eq:evidence-state-dependent} preserves the sender's incentives, fix a mediated equilibrium under \(\mathcal E^2\) with reporting matrix \(S^2\), and choose \(S^1,K\) as in that condition. Keep the receiver's original strategy. After receiving \(m\), the mediator draws \(n\) according to \(K\) and then draws a signal from the original distribution of signals following \(n\). The equality \(S^2=S^1K\) preserves the outcome and hence each type's expected payoff from following \(S^1\). The inclusion in \eqref{eq:evidence-state-dependent} ensures that every message \(n\) generated from \(m\) is available under \(\mathcal E^2\) to every type who can submit \(m\), including types who do not submit it under \(S^1\). For any such type, each of these messages gives at most his equilibrium expected payoff under \(\mathcal E^2\); a mixture of their resulting distributions of actions cannot give him a strictly higher expected payoff. Thus, the sender's prescribed strategy remains optimal.

Sequential equilibrium also requires the receiver's beliefs after unused signals to be consistent with small probabilities of sender deviations. In the original equilibrium under \(\mathcal E^2\), signals following different off-path messages can induce different beliefs because types may submit those messages with different relative probabilities. The construction using \(K\) may instead generate these signals after a single off-path message under \(\mathcal E^1\). If the signals arise only from that message, the receiver must have the same limiting belief after each of them, because the mediator's randomization conditional on a given message does not depend on the sender's type. The actions she originally chose after these signals need not all be optimal at this common belief. Thus, preserving the outcome and the sender's incentives does not by itself establish a mediated equilibrium. %

Compatible action rankings resolve this difficulty by providing a common minimizing action among those recommended with positive probability. The mediator can recommend that action after off-path messages while preserving the outcome and every type's incentives. Since the recommendation already occurs with positive probability in equilibrium, the receiver's belief follows Bayes' rule, and obedience ensures that following the recommendation is optimal for her.
\begin{theorem}\label{thm:evidence-state-dependent}
The following statements are equivalent:
\begin{enumerate}[label=\textup{(\roman*)},nosep]
\item\label{evidence-general-condition} \(\mathcal E^1\succeq^{\mathrm{G}}\mathcal E^2\).
\item\label{evidence-general-outcomes} For every set of actions \(A\), every sender payoff \(v\colon A\times\Theta\to\R\) satisfying \eqref{eq:sd-compatible-rankings}, every receiver payoff \(u\), and every full-support prior \(\mu\), every mediation outcome under evidence structure \(\mathcal E^2\) is also a mediation outcome under evidence structure \(\mathcal E^1\).
\end{enumerate}
The equivalence remains true if, in statement \ref{evidence-general-outcomes}, we fix any one full-support prior, set \(u\equiv0\), and restrict attention to sender payoffs under which all types have the same strict ranking of actions. Either statement implies \(V_1^{\mathrm{MD}}(\mu,u,v)\ge V_2^{\mathrm{MD}}(\mu,u,v)\) for every \(A,v,u,\mu\) in statement \ref{evidence-general-outcomes}.
\end{theorem}

To prove that statement (ii) implies \(\mathcal E^1\succeq^{\mathrm{G}}\mathcal E^2\), we consider sender payoffs under which all types strictly rank actions in the same order but evaluate lotteries differently. The proof assigns to every message \(n\in M^2\) a lottery over actions that gives expected payoff zero to types who can submit \(n\) and a strictly positive expected payoff to all other types. Every available message therefore gives its sender expected payoff zero, making any reporting matrix \(S^2\in\Sigma_2(\Theta)\) optimal. Reproducing the resulting outcome under \(\mathcal E^1\) yields \(S^2=S^1K\). Moreover, whenever \(K_{mn}>0\), every type who can submit \(m\) under \(\mathcal E^1\) must also be able to submit \(n\) under \(\mathcal E^2\). Otherwise, a type who can submit \(m\) but not \(n\) could obtain a strictly positive expected payoff by submitting \(m\), exceeding his equilibrium expected payoff of zero. Thus, \(K_{mn}>0\) implies \(E_m^1\subseteq E_n^2\), which, together with \(S^2=S^1K\), establishes \eqref{eq:evidence-state-dependent}.

A second case permits types to rank every pair of actions in opposite directions. Suppose each type's payoff is an affine transformation of the same function \(x\) of the action:
\begin{equation}\label{eq:evidence-affine}
v(a,\theta)=\alpha(\theta)+\beta(\theta)x(a),\qquad \text{with } x\colon A\to\R,
\end{equation}
where \(\beta(\theta)\) may be positive, negative, or zero. The sender's payoff is strictly increasing in \(x(a)\) when \(\beta(\theta)>0\) and strictly decreasing when \(\beta(\theta)<0\).

\begin{proposition}\label{prop:evidence-affine}
Suppose \eqref{eq:evidence-affine} holds and \(\mathcal E^1\succeq^{\mathrm{G}}\mathcal E^2\). For every receiver payoff \(u\) and every full-support prior \(\mu\), every mediation outcome under evidence structure \(\mathcal E^2\) is also a mediation outcome under evidence structure \(\mathcal E^1\). In particular, \(V_1^{\mathrm{MD}}(\mu,u,v)\ge V_2^{\mathrm{MD}}(\mu,u,v)\).
\end{proposition}

As discussed above, consistency can require the receiver to hold the same belief after several unused signals in the construction under \(\mathcal E^1\). Her actions following those signals under \(\mathcal E^2\) need not all be optimal at that belief. Under \eqref{eq:evidence-affine}, any two distributions of actions with the same expected value of \(x(a)\) give every sender type the same expected payoff, even when types rank actions in opposite directions. This allows the mediator to change the distribution of actions following off-path messages, while keeping the expected value of \(x(a)\) unchanged, to make the receiver willing to follow the resulting recommendations at consistent beliefs without creating profitable sender deviations.

\bibliography{event_provability_references.bib}

\appendix
\section{Proofs}\label{app:proofs}

In the arguments below, we specify the sender's and receiver's strategies and the receiver's belief, and verify that the sender and receiver act optimally and that the receiver's beliefs follow Bayes' rule whenever a message or recommendation occurs with strictly positive probability. In the Supplementary Appendix we prove that these strategies and beliefs are consistent.

\Cref{thm:complete} is the special case of \Cref{thm:sd-value} with \(v(a,\theta)=v(a)\), in which case compatible action rankings and \eqref{eq:sd-nondegeneracy} hold automatically. Likewise, \Cref{prop:some} follows from \Cref{prop:sd-some-improvement}. Since \(v\) is nonconstant, any \(a^+\in\argmax_{a\in A}v(a)\) and \(a^-\in\argmin_{a\in A}v(a)\) satisfy \(v(a^+)>v(a^-)\). Part \ref{sd-some-informative} shows that an informative message suffices, and part \ref{sd-some-binary} shows that it is necessary when \(|A|=2\).

\subsection{Proof of \autoref{thm:sd-value}}

\begin{proof}[Proof of \Cref{thm:sd-value}]
We first show that compatible action rankings imply common minimizing and maximizing actions.

\begin{claim}\label{claim:sd-common-extrema}
If sender types have compatible action rankings, every nonempty subset \(B\subseteq A\) contains an action that minimizes \(v(\cdot,\theta)\) over \(B\) for every type and an action that maximizes \(v(\cdot,\theta)\) over \(B\) for every type.
\end{claim}

\begin{proof}[Proof of the claim]
Choose \(a^-\in\argmin_{a\in B}\sum_{\theta\in\Theta}v(a,\theta)\). If some \(b\in B\) and \(\theta\) satisfied \(v(b,\theta)<v(a^-,\theta)\), compatibility would imply \(v(b,\widehat\theta)\le v(a^-,\widehat\theta)\) for every \(\widehat\theta \in \Theta\), with a strict inequality at \(\theta\).
This would contradict the choice of \(a^-\). The maximization statement follows by applying the same argument to \(-v\).
\end{proof}

Suppose statement \ref{sd-complete} holds, and fix an obedient outcome \(q\).
Let
\[
 A^+(q)\coloneqq
 \left\{a\in A\colon
 \sum_{\theta\in\Theta}\mu(\theta)q(a\mid\theta)>0
 \right\}.
\]
By \Cref{claim:sd-common-extrema}, there is an action \(a^-\in A^+(q)\) that minimizes every type's payoff over \(A^+(q)\).
Full support implies that each action assigned positive probability at any type belongs to \(A^+(q)\). Hence, \(\sum_{a\in A}v(a,\theta)q(a\mid\theta)\ge v(a^-,\theta)\) for every type. Choose an exclusive message \(m_\theta\) for each type, and use \(Y=A^+(q)\) as the signal space. Set \(g(\cdot\mid m_\theta)=q(\theta)\), and recommend \(a^-\) after every remaining message. Type \(\theta\) cannot submit another type's selected message, and his selected message gives him at least as much as any other available message. Thus, every type submitting his selected message is a best response. The receiver follows each recommendation, which is optimal by obedience of \(q\). This implements \(q\) using a pure sender strategy and at most \(|A|\) signals, proving \ref{sd-complete} implies \ref{sd-allob}.

For the converse, suppose statement \ref{sd-allob} holds. If compatible rankings failed, there would be actions \(a,b\) and types \(\theta^+,\theta^-\) such that
\[
 v(a,\theta^+)>v(b,\theta^+)
 \qquad \text{\&} \qquad
 v(a,\theta^-)<v(b,\theta^-).
\]
Set \(u(a,\theta)=u(b,\theta)=0\) for every type and set \(u(c,\theta)=-1\) for every \(c\notin\{a,b\}\).
The outcome \(q\) defined by \(q(\theta^-)=\delta_a\) and \(q(\theta)=\delta_b\) for every \(\theta\ne\theta^-\) is obedient. Suppose a mediated equilibrium induced it, and let \(p_0\) be the probability of \(a\) following the common message. Actions outside \(\{a,b\}\) are strictly dominated for the receiver, so the common message induces a lottery on \(a,b\). The incentive constraint of \(\theta^+\) requires \(p_0=0\), while the incentive constraint of \(\theta^-\) requires \(p_0=1\), a contradiction. Hence, sender types have compatible action rankings.

It remains to prove complete provability. Apply \Cref{claim:sd-common-extrema} to \(A\), and choose a common sender-minimizing action \(a^-\) and a common sender-maximizing action \(a^+\). Because \(v(\cdot,\theta)\) is nonconstant for every type,
\begin{equation}\label{eq:sd-strict-extrema}
v(a^+,\theta)>v(a^-,\theta)\qquad\text{for every }\theta\in\Theta.
\end{equation}
Suppose a type \(\theta^*\) has no exclusive message. Assign a distinct value \(x(a)\) to each existing action, with \(x(a^-)=0\), \(x(a^+)=1\), and \(x(c)>2\) for every \(c\notin\{a^-,a^+\}\). These assigned values specify the receiver's payoff; they do not change sender payoffs and need not follow the sender types' rankings of the remaining actions. Choose \(\varepsilon\in(0,1/2)\) such that
\[\varepsilon\mu(\theta^*)<\frac12\min_{\theta\ne\theta^*}\mu(\theta),\]
and set
\begin{equation}\label{eq:sd-exclusive-test}
u(a,\theta)=-\frac12\bigl(x(a)-t(\theta)\bigr)^2,\qquad \text{\&} \qquad t(\theta)=\begin{cases}\frac12+\varepsilon,&\theta=\theta^*,\\0,&\theta\ne\theta^*.\end{cases}
\end{equation}
Action \(a^+\) strictly dominates every action outside \(\{a^-,a^+\}\) for the receiver. The difference between the payoffs from the remaining two actions is
\[u(a^+,\theta)-u(a^-,\theta)=\begin{cases}\varepsilon,&\theta=\theta^*,\\-\frac12,&\theta\ne\theta^*.\end{cases}\]

Consider any mediated equilibrium, and let \(p(\theta)\) be the probability of \(a^+\) at type \(\theta\). Since every type strictly prefers \(a^+\) to \(a^-\), every message used by \(\theta^*\) induces \(a^+\) with the same probability \(p(\theta^*)\). If this probability is positive, one such message is available to a type \(\widehat\theta\ne\theta^*\), and sender optimality implies \(p(\widehat\theta)\ge p(\theta^*)\). Receiver obedience would then require
\[0\le\varepsilon\mu(\theta^*)p(\theta^*)-\frac12\sum_{\theta\ne\theta^*}\mu(\theta)p(\theta)\le\left[\varepsilon\mu(\theta^*)-\frac12\mu(\widehat\theta)\right]p(\theta^*)<0,\]
a contradiction. Thus, \(p(\theta^*)=0\), and obedience forces \(p(\theta)=0\) for every other type. Every mediated equilibrium chooses \(a^-\) at every type.

Direct disclosure attains the resulting payoff by pooling on \(m_0\), followed by \(a^-\). This action is strictly optimal at the prior because \(\varepsilon\mu(\theta^*)-\frac12[1-\mu(\theta^*)]<0\). After each unused message, concentrate belief on an eligible type other than \(\theta^*\) and choose \(a^-\). Such a type exists because \(\theta^*\) has no exclusive message. These responses are optimal and deter deviations. Full disclosure under information design instead induces \(a^+\) at \(\theta^*\) and \(a^-\) elsewhere, giving the sender strictly more. Consequently,
\begin{equation}\label{eq:sd-quadratic-gap}
V^{\mathrm D}(\mu,u,v)=V^{\mathrm{MD}}(\mu,u,v)=\sum_{\theta\in\Theta}\mu(\theta)v(a^-,\theta)<V^{\mathrm{BP}}(\mu,u,v).
\end{equation}
In particular, the full-disclosure outcome is obedient but cannot be a mediation outcome, contradicting statement \ref{sd-allob}. Hence, complete provability holds.

We have proved the equivalence of \ref{sd-complete} and \ref{sd-allob} without \eqref{eq:sd-nondegeneracy}. Since the persuasion maximum is attained, \ref{sd-allob} implies \ref{sd-bp}, and \ref{sd-bp} implies \ref{sd-improve-all}. It remains to prove \ref{sd-improve-all} implies \ref{sd-complete} under \eqref{eq:sd-nondegeneracy}. Suppose \ref{sd-improve-all} holds.

We first prove that sender types have compatible action rankings. If not, choose actions \(a,b\) that some types rank strictly in opposite directions, and write
\[
 d(\theta)\coloneqq v(a,\theta)-v(b,\theta).
\]
Condition \eqref{eq:sd-nondegeneracy} implies
\(\sum_\theta\mu(\theta)d(\theta)\ne0\).
Interchange \(a\) and \(b\) if necessary so that, with
\[
 D_+\coloneqq
 \sum_{\{\theta \colon d(\theta)>0\}}\mu(\theta)d(\theta),
 \qquad \text{\&} \qquad
 D_-\coloneqq
 \sum_{\{\theta \colon d(\theta)<0\}}\mu(\theta)[-d(\theta)],
\]
we have \(D_+>D_->0\).
Choose \(\gamma\in\left(\frac{D_-}{D_+},1\right)\), set \(u(b,\theta)=0\), set
\[
 u(a,\theta)=
 \begin{cases}
 -\gamma d(\theta),&d(\theta)>0,\\
 -d(\theta),&d(\theta)<0,\\
 0,&d(\theta)=0,
 \end{cases}
\]
and set \(u(c,\theta)=-1\) for every \(c\notin\{a,b\}\).
Action \(b\) is strictly optimal at the prior because
\[
 \sum_{\theta\in\Theta}\mu(\theta)
 [u(a,\theta)-u(b,\theta)]
 =D_--\gamma D_+<0.
\]

Consider any mediated equilibrium.
Let \(p(\theta)\) be the probability of action \(a\) at type \(\theta\), and let \(p_0\) be the probability of \(a\) following the common message.
Because \(b\) strictly dominates every action outside \(\{a,b\}\), only \(a\) and \(b\) can be chosen.
Sender optimality implies
\begin{equation}\label{eq:sd-p0-incentives}
 p(\theta)\ge p_0\quad\text{if }d(\theta)>0,
 \qquad \text{\&} \qquad
 p(\theta)\le p_0\quad\text{if }d(\theta)<0.
\end{equation}
Obedience of action \(a\) requires
\begin{equation}\label{eq:sd-a-obedience}
 \gamma
 \sum_{\{\theta \colon d(\theta)>0\}}
 \mu(\theta)d(\theta)p(\theta)
 \le
 \sum_{\{\theta \colon d(\theta)<0\}}
 \mu(\theta)[-d(\theta)]p(\theta).
\end{equation}
Combining \eqref{eq:sd-p0-incentives} and \eqref{eq:sd-a-obedience} yields \(\gamma D_+p_0\le D_-p_0\), and since \(\gamma D_+>D_-\), we have \(p_0=0\).

\eqref{eq:sd-p0-incentives} then implies \(p(\theta)=0\) whenever \(d(\theta)<0\), and \eqref{eq:sd-a-obedience} implies \(p(\theta)=0\) whenever \(d(\theta)>0\).
Types with \(d(\theta)=0\) receive the same payoff from the two actions. Every mediated equilibrium therefore gives the sender \(W_b\coloneqq\sum_{\theta\in\Theta}\mu(\theta)v(b,\theta)\), attained by the mediator always recommending \(b\).

Direct disclosure also attains \(W_b\). Let all types submit \(m_0\), after which the receiver chooses \(b\). After an unused message that is available to a type with \(d(\theta)\ge0\), concentrate the receiver's belief on such a type and choose \(b\). If every type who can submit that message has \(d(\theta)<0\), choose \(a\), which is then strictly optimal for the receiver. Every type who can submit one of the latter messages strictly prefers \(b\) to \(a\), so no deviation is profitable. Hence, \(V^{\mathrm D}(\mu,u,v)=V^{\mathrm{MD}}(\mu,u,v)=W_b\).

Persuasion delivers strictly more.
Let
\[
 \alpha\coloneqq\frac{D_-}{\gamma D_+}\in(0,1),
\]
and recommend \(a\) with probability
\[
 p(\theta)=
 \begin{cases}
 \alpha,&d(\theta)>0,\\
 1,&d(\theta)<0,\\
 0,&d(\theta)=0.
 \end{cases}
\]
The receiver is indifferent between \(a\) and \(b\) after recommendation \(a\). The outcome assigns positive probability to \(b\) only if \(d(\theta)\ge0\). The receiver strictly prefers \(b\) when \(d(\theta)>0\) and is indifferent when \(d(\theta)=0\). The sender's gain over constant action \(b\) is
\[
 \alpha D_+-D_-
 =D_-\left(\frac1\gamma-1\right)>0,
\]
contradicting statement \ref{sd-improve-all}. Hence, sender types have compatible action rankings.

Compatible action rankings have now been established. If complete provability failed, the quadratic construction \eqref{eq:sd-exclusive-test} used above would give \eqref{eq:sd-quadratic-gap}, contradicting statement \ref{sd-improve-all}. Therefore, complete provability holds.
\end{proof}

\subsection{Proof of \autoref{cor:sd-all-priors}}

\begin{proof}
Throughout this proof, statements \ref{sd-complete}--\ref{sd-improve-all} refer to \Cref{thm:sd-value}, with the requirements in \ref{sd-allob}--\ref{sd-improve-all} imposed at every full-support prior. If statement \ref{sd-complete} holds, the construction in the proof of \Cref{thm:sd-value} implements every obedient outcome at every full-support prior without using \eqref{eq:sd-nondegeneracy}, proving statement \ref{sd-allob}. Statement \ref{sd-allob} implies \ref{sd-bp} because the persuasion maximum is attained, and \ref{sd-bp} immediately implies \ref{sd-improve-all}.

Suppose statement \ref{sd-improve-all} holds. If compatible action rankings failed, choose actions \(a,b\) that some types rank in opposite directions and write \(d(\theta)=v(a,\theta)-v(b,\theta)\). There is a full-support prior \(\mu\) for which \(\sum_{\theta\in\Theta}\mu(\theta)d(\theta)\ne0\), because \(d\) takes both positive and negative values and so this equality cannot hold at every full-support prior. After relabeling the actions so that the prior-average difference is positive, the conflicting-rankings construction in the proof of \Cref{thm:sd-value} gives a receiver payoff for which \(V^{\mathrm{MD}}(\mu,u,v)=V^{\mathrm D}(\mu,u,v)<V^{\mathrm{BP}}(\mu,u,v)\). That construction needs a nonzero prior-average difference only for the selected pair of actions, not \eqref{eq:sd-nondegeneracy} for all pairs. This contradicts statement \ref{sd-improve-all}. Thus, sender types have compatible action rankings.

If complete provability failed, the quadratic construction \eqref{eq:sd-exclusive-test} could be applied at any full-support prior, with \(\varepsilon\) chosen for that prior. Equation \eqref{eq:sd-quadratic-gap} would then contradict statement \ref{sd-improve-all}. Therefore, complete provability holds, proving statement \ref{sd-complete}.
\end{proof}

\subsection{Proof of \autoref{prop:sd-some-improvement}}

We first establish a fact about direct disclosure.

\begin{lemma}\label{lem:sd-disclosure-rounding}
Suppose all receiver best responses are supported on two actions \(a,b\). From any direct-disclosure equilibrium, one can obtain a direct-disclosure equilibrium with a deterministic receiver strategy and greater \textit{ex ante} sender payoff. The sender's payoff may depend on his type, and some types may be indifferent among all actions.
\end{lemma}
\begin{proof}
Let \(p_m\) be the probability of \(a\) after message \(m\). For \(t\in(0,1)\) replace it by \(\mathbf 1_{\{p_m\geq t\}}\). This nondecreasing transformation preserves every type's best responses: types strictly preferring \(a\) maximize \(p_m\), types strictly preferring \(b\) minimize it, and indifferent types remain indifferent. Keep the sender strategy and beliefs fixed. Whenever \(0<p_m<1\), both \(a\) and \(b\) were receiver best responses, so either replacement is optimal. Pure receiver responses remain unchanged.

Averaging the transformed payoff over uniform \(t\in(0,1)\) produces the original payoff. Hence, some threshold yields a greater payoff.
\end{proof}

\begin{proof}[Proof of \Cref{prop:sd-some-improvement}]
\emph{Part \ref{sd-some-informative}.} Let \(S=E_{m^*}\) be a nonempty proper subset of the type space. Set \(u(a^-,\theta)=0\), \(u(c,\theta)=-2\) for \(c\notin\{a^-,a^+\}\), and \(u(a^+,\theta)=\varepsilon\) on \(S\) and \(-1\) on \(S^c\), where
\[
 0<\varepsilon\mu(S)<\min_{\theta\notin S}\mu(\theta).
\]
By \Cref{lem:sd-disclosure-rounding}, it is enough to bound disclosure equilibria with deterministic receiver actions. A type outside \(S\) receiving \(a^+\) with positive probability would receive it with probability one, because he strictly prefers it. Obedience would then imply \(0\leq\varepsilon\mu(S)-\mu(\theta)<0\). Hence, disclosure is bounded by
\[
 W_S=\sum_{\theta\in S}\mu(\theta)v(a^+,\theta)
 +\sum_{\theta\notin S}\mu(\theta)v(a^-,\theta).
\]
This is attained by having \(S\) submit \(m^*\) and \(S^c\) submit \(m_0\). After an unused message \(m\) with \(E_m\subseteq S\), assign any belief supported on \(E_m\) and choose \(a^+\); otherwise assign probability one to some type in \(E_m\setminus S\) and choose \(a^-\).

For mediation, send \(a^+\) after \(m^*\) and recommend \(a^+\) with probability \(\alpha\) after every other message, where
\[
 0<\alpha<\frac{\varepsilon\mu(S)}{\mu(S^c)}<1.
\]
Types in \(S\) submit \(m^*\), and the other types submit \(m_0\). Recommendation \(a^+\) is strictly optimal because \(\varepsilon\mu(S)-\alpha\mu(S^c)>0\), and recommendation \(a^-\) occurs only outside \(S\). The sender's gain is
\[
 \alpha\sum_{\theta\notin S}\mu(\theta)
 [v(a^+,\theta)-v(a^-,\theta)]>0.
\]

\emph{Part \ref{sd-some-binary}.} Let \(A=\{a,b\}\). All messages are available to all types. In a mediated equilibrium let \(p_m\) be the final probability of \(a\) after message \(m\), and define \(L=\min_m p_m\) and \(H=\max_m p_m\). Types strictly preferring \(a\) obtain \(H\), types strictly preferring \(b\) obtain \(L\), and indifferent types obtain some \(p(\theta)\in[L,H]\).

If \(H=L\), the induced outcome is constant. Obedience implies that its lottery is a receiver best response at the prior. Direct disclosure reproduces it by assigning this lottery after every message, using the prior as the belief at every message.

Suppose \(H>L\) and define \(z(\theta)=[p(\theta)-L]/(H-L)\). This outcome assigns \(a\) to types who strictly prefer \(a\), assigns \(b\) to types who strictly prefer \(b\), and may randomize at indifferent types. It  increases every type's payoff. To check obedience, put
\[
 R=\sum_\theta\mu(\theta)[u(a,\theta)-u(b,\theta)],\qquad
 F=\sum_\theta\mu(\theta)p(\theta)[u(a,\theta)-u(b,\theta)].
\]
The original obedience inequalities are \(F\geq0\) and \(F\geq R\). For the rescaled outcome the corresponding expression is
\[
 F_z=\frac{F-LR}{H-L}.
\]
If \(R\geq0\), then \(F\geq R\geq HR\) and \(F\geq LR\); if \(R\leq0\), then \(F\geq0\geq HR,LR\). Thus, \(F_z\geq0\) and \(F_z\geq R\), proving obedience.

Since \(H>L\), at least two distinct messages exist. Implement \(z\) directly using two universally available messages followed by \(a\) and \(b\): strict types choose their preferred action's message, while indifferent types randomize with probabilities \(z(\theta)\). If both messages occur, their Bayesian posteriors make the receiver's actions optimal. Extra unused messages can use a reached message's posterior and its optimal action. No strict type can strictly improve on his favorite action. If only one action occurs under \(z\), instead assign that prior-optimal action after every message. Therefore, every mediated payoff is  dominated by a direct-disclosure payoff.
\end{proof}

\subsection{Proof of \autoref{prop:sd-two-exclusive}}

\begin{proof}
Set \(u(b,\theta)=0\), \(u(c,\theta)=-2\) for \(c\notin\{a,b\}\), and
\[
 u(a,\theta)=
 \begin{cases}\varepsilon,&\theta=\theta_1,\\-1,&\theta\ne\theta_1,\end{cases}
 \qquad
 0<\varepsilon\mu(\theta_1)<\min_{\theta\ne\theta_1}\mu(\theta).
\]
Put \(d(\theta)=v(a,\theta)-v(b,\theta)\). By \Cref{lem:sd-disclosure-rounding}, one can bound disclosure using deterministic receiver strategies. Any type outside \(\theta_1\) with \(d(\theta)>0\) who received \(a\) would receive it surely, contradicting obedience by the chosen bound on \(\varepsilon\). Types with \(d(\theta)<0\) can only lower the sender's payoff relative to \(b\) if assigned \(a\); indifferent types have no effect. Consequently,
\[
 V^{\mathrm D}(\mu,u,v)\leq W=
 \mu(\theta_1)v(a,\theta_1)
 +\sum_{\theta\ne\theta_1}\mu(\theta)v(b,\theta).
\]
The bound is attained: have \(\theta_1\) send his exclusive message and all other types send \(m_0\). Choose \(a\) after the exclusive message and \(b\) at \(m_0\). For unused messages, choose \(b\) with a belief on a type other than \(\theta_1\) whenever such a type can submit the message; otherwise choose \(a\).

For mediation, map \(\theta_1\)'s selected exclusive message to \(a\), map \(\theta_2\)'s selected exclusive message to \(\alpha\delta_a+(1-\alpha)\delta_b\), and map all other messages to \(b\), where
\[
 0<\alpha<\frac{\varepsilon\mu(\theta_1)}{\mu(\theta_2)}<1.
\]
The first two types send their selected exclusive messages, and every other type sends \(m_0\). Sender incentives hold because the only messages not yielding \(b\) are exclusive to types who strictly prefer \(a\). The receiver strictly follows \(a\) because \(\varepsilon\mu(\theta_1)-\alpha\mu(\theta_2)>0\), and strictly follows \(b\) because only types other than \(\theta_1\) receive that recommendation. The payoff is
\[
 W+\alpha\mu(\theta_2)[v(a,\theta_2)-v(b,\theta_2)]>W.
\]
\end{proof}

\subsection{Proof of \autoref{cor:sd-complete-binary}}

We first establish a fact that does not require complete provability.

\begin{lemma}\label{lem:sd-opposite-binary}
With two types and two actions, if the types strictly prefer opposite actions, mediation cannot strictly improve on disclosure under any evidence structure or full-support prior.
\end{lemma}
\begin{proof}
Label types so that type 1 strictly prefers \(a\) and type 2 strictly prefers \(b\). In any mediated equilibrium, the common message imposes \(p_1\geq p_0\geq p_2\), hence \(p_1\geq p_2\). Put \(x=\mu(\theta_1)[u(a,\theta_1)-u(b,\theta_1)]\) and \(y=\mu(\theta_2)[u(a,\theta_2)-u(b,\theta_2)]\).

First suppose \(x\geq0\) and \(y\leq0\). If there are at least two messages, the two types can choose distinct messages: the type who can submit a message other than \(m_0\) uses it, and the other type uses \(m_0\). Direct disclosure then assigns each type his favorite action, which is also receiver optimal at that type. All sender types already attain their maximal payoff, so unused messages cannot induce a profitable deviation; assign any receiver best response at a belief supported on types who can submit that message. If there is only one message, mediation is constant and is reproduced by disclosure. Thus, there is no gain in this case.

If \(x<0\) and \(y\leq0\), obedience of \(a\) and \(p_1\geq p_2\) force \(p_1=p_2=0\); direct disclosure can always choose \(b\). If \(x\geq0\) and \(y>0\), obedience of \(b\) and \(p_1\geq p_2\) force \(p_1=p_2=1\); direct disclosure can always choose \(a\).

It remains to consider \(x<0<y\). Let \(R=x+y\) and \(F=xp_1+yp_2\). Obedience requires \(F\geq\max\{0,R\}\). The ordering gives \(F\leq p_2R\). If \(R>0\), these inequalities force \(p_1=p_2=1\). If \(R<0\), they force \(p_1=p_2=0\). If \(R=0\), they force \(p_1=p_2\). Thus, only constant outcomes obedient at the prior can occur. Reproduce any such constant probability \(c\) by pooling on \(m_0\). After another universally available message use the same prior-optimal lottery. After a message exclusive to type 1 choose \(b\), and after a message exclusive to type 2 choose \(a\). These are the receiver's strictly preferred actions at those types, and each is worse for the sender type who can submit the message than the pooling lottery. Hence, disclosure reproduces the mediated payoff.
\end{proof}

\begin{proof}[Proof of \Cref{cor:sd-complete-binary}]
Sufficiency follows from \Cref{prop:sd-two-exclusive}. If no two types strictly prefer the same action, nonconstancy and the binary action set imply that there are at most two types. With one type the receiver already knows the sender's type and persuasion cannot strictly improve on the sender's preferred receiver best response. With two types the rankings are opposite and \Cref{lem:sd-opposite-binary} applies. The last statement follows because at least two among three strict binary preferences must coincide.
\end{proof}

\subsection{Proof of \autoref{thm:evidence-order}}

\begin{proof}
We first establish two facts about mediation outcomes and reporting matrices.

\begin{claim}\label{claim:evidence-compact}
For fixed \(A,v,u,\mu\) with type-independent sender payoffs, and an evidence structure, the set of mediation outcomes is compact, and the highest mediated sender payoff is attained.
\end{claim}

\begin{proof}[Proof of the claim]
Let \(\ell_m\in\Delta(A)\) be the distribution of the receiver's action after the sender submits \(m\), so \(\ell_m(a)=\sum_{y\in Y}g(y\mid m)r(a\mid y)\). Every mediated equilibrium determines a reporting strategy \(\sigma\), lotteries \((\ell_m)_{m\in M}\), and an obedient outcome satisfying
\begin{equation}\label{eq:evidence-representation}
\begin{gathered}
 q(a\mid\theta)=\sum_{m\in M}\sigma(m\mid\theta)\ell_m(a),\\
 \text{\&} \qquad
 \sum_a v(a)q(a\mid\theta)\ge\sum_a v(a)\ell_m(a)
 \quad\left(m\in M(\theta)\right).
\end{gathered}
\end{equation}
Conversely, suppose \(\sigma(m\mid\theta)=0\) whenever \(m\notin M(\theta)\), and suppose \eqref{eq:evidence-representation} and \eqref{eq:obedience} hold. Let
\[
 A^+\coloneqq\left\{a\colon\sum_\theta\mu(\theta)q(a\mid\theta)>0\right\},
 \qquad \text{\&} \qquad a^-\in\argmin_{a\in A^+}v(a).
\]
Every message used by some type assigns probability one to \(A^+\). Replace the lottery at every unused message by \(\delta_{a^-}\). Each prescribed lottery gives at least \(v(a^-)\), so the replacement preserves sender optimality. Use \(A^+\) as the signal space and recommend the action drawn from the assigned lottery. Every signal occurs with positive probability, and \eqref{eq:obedience} ensures that the receiver follows it. This implements \(q\).

The reporting probabilities, message lotteries, and outcomes range over finite products of simplices. The requirements that \(\sigma(m\mid\theta)=0\) whenever \(m\notin M(\theta)\), and that \eqref{eq:evidence-representation} and \eqref{eq:obedience} hold, define a closed subset of this product. This set is compact, as is its projection onto \(q\). The sender's objective is continuous, so its maximum is attained.
\end{proof}

\begin{claim}\label{claim:evidence-message-reduction}
Fix an evidence structure and a nonempty provable set \(B\). There is a collection \(N_B\subseteq M_B\), with \(|N_B|\le|B|\), such that every reporting matrix in \(\Sigma(B)\) is a garbling of a reporting matrix using only \(N_B\). Among reporting matrices using \(N_B\), those with linearly independent columns are dense.
\end{claim}

\begin{proof}[Proof of the claim]
Form the bipartite graph with types in \(B\) on one side and messages in \(M_B\) on the other, joining \(\theta\) to \(m\) when \(\theta\in E_m\). Choose a matching with the largest possible number of pairs, and let \(N_B\) be its matched messages.

Starting from all unmatched types, repeatedly add every neighboring message and then, if that message is matched, its matched type. Let \(C\) be the types reached and \(D\) the messages reached. No unmatched message is reached: otherwise an alternating path from an unmatched type to an unmatched message would strictly increase the number of matched pairs. Thus, \(D\subseteq N_B\), every message available to a member of \(C\) belongs to \(D\), and every type outside \(C\) has a distinct matched message unavailable to \(C\). Denote the latter message by \(m_\theta\). The matched messages are precisely \(D\) and the messages \(m_\theta\) for \(\theta\notin C\).

Fix \(S\in\Sigma(B)\). For \(\theta\in C\), retain his row of \(S\), which is supported on \(D\). For \(\theta\notin C\), have him submit \(m_\theta\). Garble each message in \(D\) into its original label, and garble \(m_\theta\) into the distribution \(S(\theta)\). The assignments concern disjoint messages and reproduce every row of \(S\).

Every type has a message in \(N_B\). Have each matched type submit his matched message, and assign any available message in \(N_B\) to each remaining type. This defines a reporting matrix \(R\) whose matched rows form an identity submatrix. Given any reporting matrix \(T\) on \(N_B\), the same determinant minor of \((1-\varepsilon)T+\varepsilon R\) is a polynomial in \(\varepsilon\) that equals one at \(\varepsilon=1\). It has only finitely many zeros. Taking positive \(\varepsilon\) tending to zero outside those zeros proves the density claim.
\end{proof}

We next prove that statement \ref{evidence-order-payoff} of \Cref{thm:evidence-order}, when required at only one full-support prior, implies \eqref{eq:evidence-order}.

\begin{claim}\label{claim:evidence-payoff-necessity}
Fix a full-support prior \(\mu\). If \(V_1^{\mathrm{MD}}(\mu,u,v)\ge V_2^{\mathrm{MD}}(\mu,u,v)\) for every finite \(A\), every type-independent \(v\), and every \(u\), then \(\mathcal E^1\succeq\mathcal E^2\).
\end{claim}

\begin{proof}[Proof of the claim]
Fix a nonempty set \(B\) provable under \(\mathcal E^2\), and choose \(N_B=\{n_1,\ldots,n_r\}\) as in \Cref{claim:evidence-message-reduction}. First take a reporting matrix \(S\) on \(N_B\) with linearly independent columns. For each \(j\), define
\[
 \lambda_j^0\coloneqq\sum_{\theta\in B}\mu(\theta)S_{\theta j},
 \qquad \text{\&} \qquad
 \pi_j(\theta)\coloneqq
 \begin{cases}
 \mu(\theta)S_{\theta j}/\lambda_j^0,&\theta\in B,\\
 0,&\theta\notin B.
 \end{cases}
\]
All \(\lambda_j^0\) are positive, and the posterior vectors \(\pi_1,\ldots,\pi_r\) are linearly independent.

Introduce actions \(a_1,\ldots,a_r\) and \(b_{jt}\), where \(j=1,\ldots,r\) and \(t\in\Theta\). Set \(v(a_j)=1\) and \(v(b_{jt})=0\) for every \(j\) and \(t\). For the receiver, set
\[u(a_j,\theta) = 2\pi_j(\theta)-\|\pi_j\|^2,\qquad \text{\&} \qquad u(b_{jt},\theta) =u(a_j,\theta)
       +\varepsilon\left(\1_{\{\theta=t\}}-\pi_j(t)\right),\]
where the norm is Euclidean and \(\varepsilon>0\) is small enough that
\[
 \varepsilon|\pi_j(t)-\pi_i(t)|<\|\pi_j-\pi_i\|^2
 \qquad(i\ne j,\ t\in\Theta).
\]
When \(r=1\), any positive \(\varepsilon\) suffices.

At posterior \(p\), the expected payoff difference between \(b_{jt}\) and \(a_j\) is \(\varepsilon[p(t)-\pi_j(t)]\). If \(p\ne\pi_j\), some coordinate has \(p(t)>\pi_j(t)\), so \(a_j\) is not optimal. At \(\pi_j\), every \(b_{jt}\) ties with \(a_j\), while
\[
 \E_{\pi_j}[u(a_j,\theta)-u(a_i,\theta)]
 =\|\pi_j-\pi_i\|^2>0
 \qquad(i\ne j).
\]
The choice of \(\varepsilon\) also makes \(a_j\) strictly preferable to every \(b_{it}\) with \(i\ne j\). Thus \(a_j\) is optimal exactly at \(\pi_j\). At every posterior there is a receiver best response among the actions \(b_{jt}\): for each \(a_j\), at least one of its associated actions \(b_{jt}\) gives the receiver at least as much.

Under \(\mathcal E^2\), assign action \(a_j\) to message \(n_j\), and have types in \(B\) report according to \(S\). If \(B\ne\Theta\), assign every remaining message a receiver-optimal action among the \(b_{jt}\) at \(\mu(\cdot\mid\Theta\setminus B)\), and have all types outside \(B\) submit the common message. If \(B=\Theta\), assign all remaining messages \(a_1\). Members of \(B\) obtain payoff one, and outsiders, if present, obtain zero and cannot submit any \(n_j\). Recommendation \(a_j\) induces posterior \(\pi_j\), so this is an equilibrium with sender payoff \(\mu(B)\).

Under either evidence structure, action \(a_j\) can be optimal only at posterior \(\pi_j\), which assigns probability one to \(B\). Hence, no mediated payoff exceeds \(\mu(B)\). The assumed inequality \(V_1^{\mathrm{MD}}(\mu,u,v)\ge V_2^{\mathrm{MD}}(\mu,u,v)\) and \Cref{claim:evidence-compact} imply that \(\mathcal E^1\) has an equilibrium attaining \(\mu(B)\). Every type in \(B\) must then receive an action in \(\{a_1,\ldots,a_r\}\) with probability one, and every type outside \(B\) receives payoff zero.

Let \(\lambda_j\) be the unconditional probability of action \(a_j\) in this equilibrium. Whenever it occurs, its conditional type distribution is \(\pi_j\). Bayes' rule therefore implies
\[
 \sum_{j=1}^r\lambda_j\pi_j(\theta)
 =\mu(\theta)\1_{\{\theta\in B\}}
 =\sum_{j=1}^r\lambda_j^0\pi_j(\theta)
 \qquad(\theta\in\Theta).
\]
Linear independence yields \(\lambda_j=\lambda_j^0\) for every \(j\). Thus, the conditional probability of \(a_j\) at every \(\theta\in B\) equals \(S_{\theta j}\).

Every message used by a member of \(B\) must assign probability one to \(\{a_1,\ldots,a_r\}\). Such a message must rule out every type outside \(B\); otherwise an outsider could deviate from payoff zero to one. Restrict the equilibrium reporting strategy to \(B\) to obtain \(S^1\in\Sigma_1(B)\). For each message used by a member of \(B\), let \(K_{mj}\) be the probability of final action \(a_j\) following that message. Complete unused rows of \(K\) arbitrarily to probability distributions. Then \(S=S^1K\).

We have established the factorization for every reporting matrix on \(N_B\) with linearly independent columns. These matrices are dense by \Cref{claim:evidence-message-reduction}. The possible \(S^1\) and \(K\) range over compact sets, so taking limits establishes the factorization for every reporting matrix on \(N_B\). Finally, every element of \(\Sigma_2(B)\) is a garbling of such a matrix. Composing the two garblings proves \eqref{eq:evidence-order}. %
\end{proof}

\begin{claim}\label{claim:evidence-outcome-sufficiency}
If \(\mathcal E^1\succeq\mathcal E^2\), every mediation outcome under \(\mathcal E^2\) is also a mediation outcome under \(\mathcal E^1\), for every \(A,v,u\) with type-independent \(v\) and every full-support \(\mu\).
\end{claim}

\begin{proof}[Proof of the claim]
Fix an equilibrium under \(\mathcal E^2\), with reporting strategy \(\sigma^2\), final-action lotteries \(\ell_n\) following its messages, and induced outcome \(q\). Write
\[
 z_n\coloneqq\sum_a v(a)\ell_n(a),
 \qquad \text{\&} \qquad
 w(\theta)\coloneqq\sum_a v(a)q(a\mid\theta)
 =\max_{n\in M^2(\theta)}z_n.
\]
For each value \(c\in w(\Theta)\), let \(B_c\coloneqq\{\theta \colon w(\theta)\ge c\}\). If \(n\) is used by a type in \(B_c\), sender optimality implies \(z_n\ge c\). If \(n\) were available to some \(\widehat\theta\notin B_c\), that type could obtain \(z_n\ge c>w(\widehat\theta)\) by submitting it, contradicting sender optimality. The restriction of \(\sigma^2\) to \(B_c\) is, consequently, an element \(S^{2,c}\in\Sigma_2(B_c)\). By \eqref{eq:evidence-order}, choose \(S^{1,c}\in\Sigma_1(B_c)\) and \(K^c\) with \(S^{2,c}=S^{1,c}K^c\).

Let \(J_c\coloneqq \bigcup_{\{\theta \colon w(\theta)=c\}}\{m\colon S^{1,c}_{\theta m}>0\}\). If \(m\in J_c\), choose \(\theta\) with \(w(\theta)=c\) and \(S^{1,c}_{\theta m}>0\). Since
\[
 \sigma^2(n\mid\theta)
 =
 \sum_{\widehat m\in M_{B_c}^1}
 S^{1,c}_{\theta\widehat m}K^c_{\widehat m n}
 \ge
 S^{1,c}_{\theta m}K^c_{mn},
\]
\(K^c_{mn}>0\) implies \(\sigma^2(n\mid\theta)>0\), and hence \(z_n=w(\theta)=c\). Assign message \(m\) the lottery \(\ell_m^1\coloneqq\sum_{n\in M_{B_c}^2}K^c_{mn}\ell_n\), yielding payoff \(c\). Moreover, \(E_m^1\subseteq B_c\) and \(m\) is available to a type receiving \(c\), so \(\min_{\theta\in E_m^1}w(\theta)=c\). Thus, the sets \(J_c\) are disjoint, so no message is assigned more than one lottery.

For a type \(\theta\) with \(w(\theta)=c\), use the row \(S^{1,c}(\theta)\) as his reporting strategy. Its support is contained in \(J_c\), and
\[
 \sum_{m\in J_c}S^{1,c}_{\theta m}\ell_m^1
 =\sum_{n\in M_{B_c}^2}\sigma^2(n\mid\theta)\ell_n
 =q(\theta).
\]
Every assigned message delivering payoff \(c\) rules out every type outside \(B_c\), so any type who can submit it already receives at least \(c\). Assign every remaining message the sender's least-preferred action among those used by \(q\). No deviation is profitable. All assigned lotteries are supported on actions used by \(q\); recommend those actions and have the receiver follow them. The induced outcome is still \(q\), so receiver optimality follows from \eqref{eq:obedience}, and every signal occurs with positive probability.
\end{proof}

\Cref{claim:evidence-outcome-sufficiency} proves \ref{evidence-order-condition} implies \ref{evidence-order-outcomes}, and \ref{evidence-order-outcomes} implies \ref{evidence-order-payoff} immediately. \Cref{claim:evidence-payoff-necessity} proves the remaining implication using any one full-support prior. Since \eqref{eq:evidence-order} contains no prior, its sufficiency applies at every full-support prior, establishing the fixed-prior observation following the theorem.
\end{proof}

\subsection{Proof of \autoref{thm:evidence-H}}

\begin{proof}
We first show that, under \eqref{eq:evidence-H}, an obedient outcome can be induced exactly when \eqref{eq:evidence-upper-sets} holds.

\begin{claim}\label{claim:evidence-upper-characterization}
Every mediation outcome has provable upper contour sets of sender payoffs. Under \eqref{eq:evidence-H}, every obedient outcome with this property can be induced using a pure sender strategy and at most \(|A|\) signals.
\end{claim}

\begin{proof}[Proof of the claim]
In any equilibrium, let \(z_m\) be the sender's payoff following message \(m\). Sender optimality implies
\begin{equation}\label{eq:evidence-payoff-unions}
 w(\theta)=\max_{m\in M(\theta)}z_m,
 \qquad \text{\&} \qquad
 \{\theta \colon w(\theta)\ge c\}
 =\bigcup_{\{m \colon z_m\ge c\}}E_m.
\end{equation}
This proves necessity without \eqref{eq:evidence-H}.

For sufficiency, suppose \eqref{eq:evidence-H} holds and \(q\) is obedient with provable upper contour sets. For each \(c\in w(\Theta)\), let \(B_c=\{\theta \colon w(\theta)\ge c\}\) and choose the injection \(f_{B_c}\) in \eqref{eq:evidence-H}. Assign to type \(\theta\) the message \( m_\theta\coloneqq f_{B_{w(\theta)}}(\theta)\). Types with the same value of \(w(\theta)\) receive distinct assigned messages by injectivity. Types with different values of \(w(\theta)\) also receive distinct assigned messages: a message that rules out every type outside \(B_d\) cannot be available to a type receiving \(c<d\).

Assign lottery \(q(\theta)\) to \(m_\theta\), and assign every other message the sender's least-preferred action among those used by \(q\). If type \(\theta\) can submit \(m_{\widehat\theta}\), then \(\theta\in E_{m_{\widehat\theta}}
 \subseteq B_{w(\widehat\theta)}\) and \(\qquad w(\theta)\ge w(\widehat\theta)\), and so each type best responds by submitting his assigned message. All other messages yield no more than his prescribed payoff. Recommend the actions drawn from the assigned lotteries. The induced outcome is \(q\), so \eqref{eq:obedience} ensures receiver optimality. Using only actions occurring under \(q\) leaves no unused signals and requires at most \(|A|\) signals.
\end{proof}

If \(\mathcal P_2\subseteq\mathcal P_1\), every mediation outcome under \(\mathcal E^2\) has upper contour sets in \(\mathcal P_1\), by \eqref{eq:evidence-payoff-unions}. Since \(\mathcal E^1\) satisfies \eqref{eq:evidence-H}, \Cref{claim:evidence-upper-characterization} implies that it induces the same outcome. This proves \ref{evidence-H-inclusion} implies \ref{evidence-H-outcomes}. Statement \ref{evidence-H-outcomes} implies \ref{evidence-H-payoff} immediately.

\begin{claim}\label{claim:evidence-event-test}
For a fixed type-independent sender payoff \(v\), universal sender-payoff dominance implies \(\mathcal P_2\subseteq\mathcal P_1\), without imposing \eqref{eq:evidence-H} on either evidence structure.
\end{claim}

\begin{proof}[Proof of the claim]
Choose \(a^-\in\argmin_{a\in A}v(a)\) and \(a^+\in\argmax_{a\in A}v(a)\). Fix a nonempty set of types \(B\) and set
\[u_B(a^-,\theta) = 0, \quad u_B(a^+,\theta) =
 \begin{cases}
 0,&\theta\in B,\\
 -1,&\theta\notin B,
 \end{cases}\quad \text{\&} \quad
 u_B(a,\theta) =-2\quad(a\notin\{a^-,a^+\}).
\]
For \(i\in\{1,2\}\), write
\(B_i\coloneqq\bigcup_{\{m\in M^i\colon E_m^i\subseteq B\}}E_m^i\). We show that
\begin{equation}\label{eq:evidence-event-value}
 V_i^{\mathrm{MD}}(\mu,u_B,v)
 =v(a^-)+\mu(B_i)\left[v(a^+)-v(a^-)\right].
\end{equation}

Only \(a^-\) and \(a^+\) can be receiver best responses, and obedience requires zero probability of \(a^+\) outside \(B\). Every outsider thus receives payoff \(v(a^-)\). A message producing \(a^+\) with positive probability must exclude all outsiders, since otherwise one could profitably submit it. Such a message is available only to types in \(B\). Hence, only members of \(B_i\) can receive \(a^+\), proving the upper bound in \eqref{eq:evidence-event-value}.

To attain the bound, recommend \(a^+\) after every message available only to types in \(B\), and recommend \(a^-\) after every other message. Members of \(B_i\) choose a message of the first kind, and all remaining types choose a message of the second kind. Sender incentives hold, and both recommendations are obedient whenever they occur. Omit any recommendation that does not occur. This proves \eqref{eq:evidence-event-value}.

If \(B\in\mathcal P_2\setminus\mathcal P_1\), then \(B_2=B\) and \(B_1\subsetneq B\). Full support and nonconstancy of \(v\) imply
\[
 V_2^{\mathrm{MD}}(\mu,u_B,v)-V_1^{\mathrm{MD}}(\mu,u_B,v)
 =\mu(B\setminus B_1)\left[v(a^+)-v(a^-)\right]>0.
\]
This contradicts universal payoff dominance.
\end{proof}

\Cref{claim:evidence-event-test} proves \ref{evidence-H-payoff} implies \ref{evidence-H-inclusion}, completing the equivalence. The final assertion follows from \Cref{claim:evidence-upper-characterization}.

To verify the sharpness statement, suppose \(\mathcal E^1\) fails \eqref{eq:evidence-H} on a nonempty provable set \(B\). Obtain \(\mathcal E^2\) by adding \(|B|\) distinct messages, each available precisely to types in \(B\). Since \(B\in\mathcal P_1\) and \(\mathcal P_1\) is closed under unions, \(\mathcal P_2=\mathcal P_1\).

Take actions \(a_0\) and \(a_\theta\) for \(\theta\in B\), with
\[
 v(a_0)=0,\qquad v(a_\theta)=1,
 \qquad u(a_0,\widehat\theta)=0,
 \qquad \& \qquad
 u(a_\theta,\widehat\theta)=
 \begin{cases}
 0,&\widehat\theta=\theta,\\
 -1,&\widehat\theta\ne\theta.
 \end{cases}
\]
A favorable action \(a_\theta\) is optimal only at the posterior assigning probability one to \(\theta\). \(\mathcal E^2\) gives every \(\theta\in B\) his associated action through a different added message, and assigns \(a_0\) to all original messages. Members of \(B\) obtain sender payoff one, and outsiders obtain zero. This is an equilibrium, and its payoff \(\mu(B)\) is the largest possible under either evidence structure.

If \(\mathcal E^1\) attained \(\mu(B)\), each \(\theta\in B\) would receive \(a_\theta\). Choose one message used by each such type. These messages must be distinct, since they induce different certain actions. They must also rule out every type outside \(B\), since otherwise an outsider could deviate from payoff zero to one. They would therefore supply the injection required by \eqref{eq:evidence-H}, a contradiction. By \Cref{claim:evidence-compact}, the optimum under \(\mathcal E^1\) is attained, so it is strictly less than \(\mu(B)\).
\end{proof}

\subsection{Proof of \autoref{thm:evidence-state-dependent}}

\begin{proof}
Suppose first that \(\mathcal E^1\succeq^{\mathrm{G}}\mathcal E^2\), and that sender types have compatible action rankings. Fix a mediated equilibrium under \(\mathcal E^2\), with reporting strategy \(S^2\), final-action lottery \(\ell_n\) following each message \(n\in M^2\), and outcome \(q\). Write
\[z_n(\theta)\coloneqq\sum_{a\in A}v(a,\theta)\ell_n(a),\qquad w(\theta)\coloneqq\sum_{a\in A}v(a,\theta)q(a\mid\theta).\]
Sender optimality gives \(z_n(\theta)\le w(\theta)\) for every \(n\in M^2(\theta)\). Choose \(S^1,K\) as in \eqref{eq:evidence-state-dependent}, and initially assign message \(m\in M^1\) the lottery
\[\ell_m^1\coloneqq\sum_{n\in M^2}K_{mn}\ell_n.\]
Then \(S^1\ell^1=q\). If \(\theta\in E_m^1\), the restriction on \(K\) ensures that every \(n\) with \(K_{mn}>0\) is available to \(\theta\) under \(\mathcal E^2\). Hence,
\[\sum_a v(a,\theta)\ell_m^1(a)=\sum_nK_{mn}z_n(\theta)\le w(\theta).\]
The prescribed mixture gives \(w(\theta)\), so every message used by \(\theta\) attains this bound.

Let \(A^+(q)\) denote the actions used by \(q\). By \Cref{claim:sd-common-extrema}, this set contains an action \(a^-\) that minimizes every type's payoff on the set. Retain \(\ell_m^1\) after every used message and assign \(\delta_{a^-}\) after every unused message. Used-message lotteries are supported on \(A^+(q)\), and every type receives at least \(v(a^-,\theta)\) under \(q\). Thus, the replacement preserves sender optimality and \(q\). Use \(A^+(q)\) as the signal space and recommend the drawn action. Obedience of \(q\) makes the receiver follow each recommendation. This proves statement \ref{evidence-general-outcomes}.

For necessity, fix a full-support prior \(\mu\) and suppose the comparison in \ref{evidence-general-outcomes} holds when \(u\equiv0\) and all types have the same strict action ranking. Fix \(S^2\in\Sigma_2(\Theta)\), and list the messages of \(\mathcal E^2\) as \(n_1,\ldots,n_r\). Use distinct actions \(a_j^-,a_j^+\), \(j=1,\ldots,r\), set \(u\equiv0\), and let
\begin{equation}\label{eq:evidence-compatible-test}
v(a_j^-,\theta)=-2j+\1_{\{\theta\notin E_{n_j}^2\}},\qquad v(a_j^+,\theta)=2j+\1_{\{\theta\notin E_{n_j}^2\}}.
\end{equation}
At every type, payoffs strictly increase along \(a_r^-,\ldots,a_1^-,a_1^+,\ldots,a_r^+\). After message \(n_j\), the mediator recommends each of \(a_j^-,a_j^+\) with probability \(1/2\). The expected sender payoff from that lottery is \(\1_{\{\theta\notin E_{n_j}^2\}}\). Every available message gives its sender payoff zero, so reporting according to \(S^2\) is optimal. The receiver follows each recommendation because \(u\equiv0\). The resulting outcome satisfies
\[q(a_j^-\mid\theta)=q(a_j^+\mid\theta)=\frac12S^2_{\theta n_j}.\]

By the assumed comparison, \(\mathcal E^1\) has a mediated equilibrium with this outcome. Let \(S^1\) be its reporting strategy and \(\ell_m^1\) the distribution of final actions after \(m\). Every type's equilibrium payoff is zero. For every used message \(m\), set
\[K_{mn_j}\coloneqq\ell_m^1(a_j^-)+\ell_m^1(a_j^+).\]
Choose a type \(\theta_0\) using \(m\). If \(n_j\notin M^2(\theta_0)\), both corresponding actions have zero probability under \(q(\theta_0)\), and hence zero probability following \(m\). The indicator terms in \eqref{eq:evidence-compatible-test} therefore vanish on the support of \(\ell_m^1\) at \(\theta_0\). Since the used message gives \(\theta_0\) payoff zero,
\[\sum_{j=1}^r2j\bigl[\ell_m^1(a_j^+)-\ell_m^1(a_j^-)\bigr]=0.\]
For any type \(\theta\in E_m^1\), his payoff from submitting \(m\) is consequently
\[\sum_a v(a,\theta)\ell_m^1(a)=\sum_{\{j:\theta\notin E_{n_j}^2\}}K_{mn_j}.\]
This nonnegative quantity cannot exceed his equilibrium payoff of zero. Thus, \(K_{mn_j}>0\) implies \(E_m^1\subseteq E_{n_j}^2\). For unused messages, complete the corresponding rows of \(K\) by assigning probability one to the common message of \(\mathcal E^2\). The required inclusion holds there as well. Summing the two action probabilities gives
\[S^2_{\theta n_j}=q(a_j^-\mid\theta)+q(a_j^+\mid\theta)=\sum_{m\in M^1}S^1_{\theta m}K_{mn_j}.\]
Therefore, \(S^2=S^1K\), proving \eqref{eq:evidence-state-dependent}. The same construction applies to every \(S^2\) at the fixed prior. Finally, reproducing an outcome preserves its \textit{ex ante} sender payoff, giving the payoff comparison.
\end{proof}

\subsection{Proof of \autoref{prop:evidence-affine}}

We first establish a fact about receiver best responses.

\begin{lemma}\label{lem:evidence-affine-mean}
Fix \(u\), a scalar function \(x\colon A\to\R\), and a lottery \(\ell\) supported on actions that maximize the receiver's payoff at some belief. There are a belief \(\nu\) and a lottery \(\widehat\ell\) supported on at most two receiver best responses at \(\nu\) such that \(\sum_a x(a)\widehat\ell(a)=\sum_a x(a)\ell(a)\).
\end{lemma}

\begin{proof}
Put \(z=\sum_a x(a)\ell(a)\). Let \(F_-\) be the beliefs at which some receiver best response has \(x(a)\le z\), and let \(F_+\) be the beliefs at which some receiver best response has \(x(a)\ge z\). Each is a finite union of closed best-response regions and is therefore closed. They cover \(\Delta(\Theta)\), and both are nonempty because \(z\) is an average of the values of actions that are optimal at some beliefs. Connectedness of \(\Delta(\Theta)\) gives \(F_-\cap F_+\ne\varnothing\). At a belief \(\nu\) in the intersection, choose best responses with values below and above \(z\). A mixture of them has expected value \(z\). If either value equals \(z\), use that action alone.
\end{proof}

\begin{proof}[Proof of \Cref{prop:evidence-affine}]
Start with a mediated equilibrium under \(\mathcal E^2\), with outcome \(q\), reporting matrix \(S^2\), and final-action lotteries \(\ell_n\). Choose \(S^1,K\) from \eqref{eq:evidence-state-dependent}, and assign \(\ell_m^1=\sum_nK_{mn}\ell_n\). The factorization and sender-incentive argument in the proof of \Cref{thm:evidence-state-dependent} give \(S^1\ell^1=q\) and
\[\sum_a v(a,\theta)\ell_m^1(a)\le w(\theta)\coloneqq\sum_a v(a,\theta)q(a\mid\theta)\qquad(\theta\in E_m^1).\]
Every action in these lotteries is optimal for the receiver at some belief, since it is chosen at a receiver information set in the original equilibrium.

Let \(\ell=\ell_{m_0^1}^1\). This lottery is available to every type. Retain \(\ell_m^1\) after used messages. If \(\supp(\ell)\subseteq A^+(q)\), assign \(\ell\) after all unused messages and recommend actions directly. The sender's strategy remains optimal, the outcome remains \(q\), and every recommendation occurs and is obedient.

Otherwise, \(m_0^1\) is unused. Apply \Cref{lem:evidence-affine-mean} to obtain \(\widehat\ell\) and \(\nu\). By \eqref{eq:evidence-affine}, for every type,
\[\sum_a v(a,\theta)\widehat\ell(a)=\alpha(\theta)+\beta(\theta)\sum_a x(a)\widehat\ell(a)=\sum_a v(a,\theta)\ell(a)\le w(\theta).\]
Assign \(\widehat\ell\) after every unused message. The sender's strategy and the outcome are unchanged. Recommend the drawn actions. Reached recommendations are obedient by \(q\), and every unused recommendation is optimal at \(\nu\). Hence, \(q\) is a mediation outcome under \(\mathcal E^1\). Comparing \textit{ex ante} payoffs gives the stated inequality.
\end{proof}

\end{document}